\documentclass[12pt,a4paper,twoside]{article}
  \usepackage{authblk} 
  \usepackage{url}
  \usepackage{amsthm}
  \usepackage{amsmath}
  \usepackage{amsfonts}
  \usepackage{amssymb}
  \usepackage{amscd}
  \usepackage{authblk}
  \usepackage[mathscr]{eucal}
  \usepackage{mathtools} 
  \usepackage{enumitem}
  \usepackage{physics}
  \usepackage{color}
  \usepackage{fancyhdr}
  \usepackage{hyperref}
 \usepackage[disable]{todonotes}
  \usepackage{marginnote}
  \usepackage[top=2.5cm, bottom=2.7cm,left=3.5cm, right=2.5cm, marginparwidth=1.8cm]{geometry}
  \usepackage{framed}
  \usepackage[noabbrev]{cleveref} 
  \usepackage[T1]{fontenc} 
 \usepackage[utf8]{inputenc}
 \usepackage{blkarray}

  \numberwithin{equation}{section}
 
 \allowdisplaybreaks[1] 
  
  \swapnumbers 

  \theoremstyle{definition}  
   \newtheorem{defn}{Definition}[section]
   \newtheorem{eg}[defn]{Example}
   \newtheorem{egs}[defn]{Examples}

   \newtheorem{rmk}[defn]{Remark}
   \newtheorem{rmks}[defn]{Remarks}

  \theoremstyle{plain}  
   \newtheorem{thm}[defn]{Theorem}
   \newtheorem{lem}[defn]{Lemma}
   \newtheorem{prop}[defn]{Proposition}
   \newtheorem{cor}[defn]{Corollary}

  \theoremstyle{remark} 

   \newtheorem*{claim}{Claim}
   \newtheorem*{pf}{Proof}

   \newcommand{\B}[1]{\mathscr{B}({#1})}
   \newcommand{\SH}[1]{\mathscr{S}({#1})}
   \newcommand{\CH}{\mathcal{H}}
   \newcommand{\GH}{\Gamma_{s}(\mathcal{H})}
 \newcommand{\numberthis}{\refstepcounter{equation}\tag{\theequation}} 
\makeatletter
\newcommand\footnoteref[1]{\protected@xdef\@thefnmark{\ref{#1}}\@footnotemark}
\makeatother 

\newcommand{\emailaddress}[1]{\newline{\sf#1}}

\renewcommand{\colon}{\nobreak\mskip2mu\mathpunct{}\nonscript
  \mkern-\thinmuskip{:}\mskip6muplus1mu\relax}
\let\tr\relax 
\DeclareMathOperator{\tr}{Tr}
\DeclareMathOperator{\diag}{Diag}

\DeclareMathOperator{\re}{Re}
\DeclareMathOperator{\im}{Im}

\DeclareMathOperator*{\slim}{s-lim}

\fontsize{12}{14}

\title{Infinite Mode Quantum Gaussian States}
\author[1]{B.V. Rajarama Bhat}
\author[2]{Tiju Cherian John}
\author[3]{R. Srinivasan}

\affil[1,2]{Indian Statistical Institute,  Stat-Math Unit, 8th Mile, Mysore Road, R.V. College Post, Bengaluru-560059, India. \emailaddress{bhat@isibang.ac.in, tijucherian@gmail.com}}
\affil[3]{Chennai Mathematical Institute, H1, SIPCOT IT Park, Siruseri, Kelambakkam, Chennai-603103, India. \emailaddress{vasanth@cmi.ac.in}}

\begin{document}
\maketitle
   \begin{abstract}
Quantum Gaussian states on Bosonic Fock spaces  are quantum versions of Gaussian distributions. In this paper we explore infinite mode quantum Gaussian states.
We extend many of the results of Parthasarathy in \cite{Par10} and \cite{Par13} to the infinite mode case, which includes various characterizations, convexity and symmetry properties.

     \textbf{Keywords:}
Quantum Gaussian states, 
Williamson's normal form, Infinite mode quantum systems.\\
\textbf{2010 Mathematics Subject classification:}
Primary 81S05; Secondary 46L60.
\end{abstract}

\section{Introduction}
There is extensive literature on quasifree states on CCR (Canonical Commutation Relation) algebras variously known as quantum Gaussian states or squeezed states. Some early references being
\cite{Arakiccr1}, \cite{Arakiccr2}, \cite{Hol71}, \cite{Hol71ii}, and \cite{vD71}. Comprehensive treatments and additional references can be seen in monographs \cite{brattelirobinson2},
\cite{Pet90} and \cite{derezinkigerard}. 

Recently finite mode Gaussian states have been getting more attention in the context of its importance in quantum information theory. Some references are \cite{AdRaSa14}, \cite{ferraro05}, \cite{Wolf06} and \cite{WANG20071}. K R Parthasarathy wrote an expository article on finite mode Gaussian states in \cite{Par10} and followed it up
with a series of articles systematically studying convexity, symmetry and dilation properties (see  \cite{Par13, Par15}).
The aim of this article is to extend some of these results to the infinite mode case. Now the covariance operators are restricted by some trace class or Hilbert-Schmidt conditions. We take care of these technicalities by developing necessary notation and tools. The infinite mode setting has been studied earlier in different contexts  by \cite{HR96,derezinkigerard} and others.

 The scheme is as follows. In Section 2, we review some terminology and known results. For notation and basics on Boson Fock spaces we follow mostly  \cite{Par12}. We need some minor refinements regarding existing results on Shale Operators and Bogoliubov transformations \cite{Sha62,Par12,BhSr05,HR96}. We crucially make use of the notion of quasi-free states on CCR Algebra and their classification theory. For this we depend on \cite{Pet90},  \cite{Hol71}, \cite{Hol71ii}  and  \cite{vD71}.  Finally, we need suitable extension of Williamson's normal form to infinite dimensions \cite{BhJo18}. The result we need has been quoted in Section 2.4. 

Section 3, has the formal definition of quantum Gaussian states in the infinite mode case, by first introducing quantum characteristic function or quantum Fourier transform of states on Boson Fock space.
Every quantum Gaussian state comes with a `covariance matrix', which in our setup is a symmetric and invertible real linear operator. Making use of the theory of quasifree states we derive some necessary conditions on this operator. Some further analysis through explicit constructions finally leads to a complete characterization of covariance operators of quantum Gaussian states in infinite mode (Theorem \ref{thm:main-thm}).  

Now we have the proper setting to extend the results of K R Parthasarathy \cite{Par13} to infinite mode Gaussian states. We characterize extreme points of the set of Gaussian state covariance operators and express every interior point as midpoint of two extreme points (Theorem \ref{sec:conv-prop-covar-1}). Theorem \ref{thm:struct-quant-gauss} provides explicit description for every quantum Gaussian state, in terms of symplectic eigenvalues of the covariance operator. We see that every mixed Gaussian state can be purified to a pure Gaussian state.

 A unitary operator in the Boson Fock space is called a Gaussian symmetry if it preserves Gaussianity of states under conjugation. The last result in this article (Theorem \ref{thm:Gaussian-symmetry}) is a complete characterization of Gaussian symmetries and shows that symplectic spectrum is a complete invariant for Gaussian states under conjugation by Gaussian symmetries.

\section{Preliminaries and Notations}

In this section we fix the notation used in this paper and recall some of the well-known results. Throughout we consider only separable Hilbert spaces.
\subsection{Symmetric Fock space}  

 Let $\CH$ be a complex Hilbert space  with inner product $\left\langle \cdot,\cdot \right\rangle$, which is anti-linear in the first variable. 
The symmetric Fock space or the  Boson Fock space over $\CH$ is defined as 
\begin{displaymath} 
\Gamma_s(\CH) := \bigoplus\limits_{n=0}^{\infty}\mathcal{H}^{\small{\textcircled{s}}^{n}}
\end{displaymath} 
where $\mathcal{H}^{\small{\textcircled{s}}^{n}}$ is the  $n$-fold symmetric tensor product of $\CH$ and  $\mathcal{H}^{\small{\textcircled{s}}^{0}}:=\mathbb{C}$. The $n$-th direct summand is called the $n$-\emph{particle subspace} \index{$n$- particle subspace}. When $n=0$, it is called the 
\emph{vacuum subspace}\index{vacuum subspace}. The vector $\Phi :=1 \oplus 0 \oplus 0 \oplus \cdots \in \GH$ is called the \emph{vacuum vector}. Any element in the $n$-particle subspace is called an \emph{$n$-particle vector}.  The dense linear manifold generated by the $n$-particle vectors, $n=0,1,2,\dots$ is called the \emph{finite particle space}, denoted by $\Gamma_s^0(\CH)$. 
For $f\in \mathcal{H}$, the exponential vector $e(f) \in \Gamma_{s}(\mathcal{H})$  is defined by
\begin{equation}
\label{eq:1.a}
e(f) = \oplus_{n=0}^\infty  
\frac{f^{\otimes^n}}{\sqrt{n!}},
\end{equation}
where $f^{\otimes^0} := 1$.
Then $\left\langle e(f), e(g) \right\rangle =  \exp \left\langle f, g \right\rangle$.  Further 
 the set $E:=\{e(f)|f\in \CH\}$ of all exponential vectors is linearly independent and  total in  $\Gamma_s(\CH)$. Indeed if $A$ is a dense set in $\mathcal{H}$, then the linear span of the set $\{e(f)|f\in A\}$ is dense in $\GH$.
\begin{eg}\label{eg:fock-L2-iso}
 We describe a well known identification of $\Gamma_s(\mathbb{C})$ here, this identification prove to be very useful for the later part of this work. One 
 can  refer Example 19.8 and Exercise 20.20 in \cite{Par12} for more details. Consider $L^2(\mathbb{R})$, where  $\mathbb{R}$ is endowed with the Lebesgue measure. Identify $e(z) \in \Gamma_s(\mathbb{C})$ with the $L^2$-function $x \mapsto (2\pi)^{-1/4}\exp{-4^{-1}x^2+zx-2^{-1}z^{2}}, x\in \mathbb{R}$. Then this identification extends as a Hilbert space isomorphism between $\Gamma_s(\mathbb{C})$ and $L^2(\mathbb{R})$. We may write $\Gamma_s({\mathbb{C}}) = L^2(\mathbb{R})$ in this sense. 
\end{eg}
The Weyl operator, corresponding to an element $f \in \mathcal{H}$, is defined on 
 the set of exponential vectors by \begin{equation}
\label{eq:7}
W(f)e(g) = \{\exp(-\frac{1}{2}\|f\|^2-\left\langle f,g \right\rangle) \}e(f+g), 
  \end{equation} extends to a unitary operator on $\GH$.
 The mapping $f\mapsto W(f)$ from $\CH$ into $\B{\Gamma_{s}(\CH)}$ is strongly continuous. Further,   
\begin{equation}\label{eq:8}
W(-f)=W(f)^{*};~~
W(f)W(g) = \exp(-i \im \left\langle f,g \right\rangle)W(f+g)~~~\forall f, g \in \CH, 
\end{equation}
where `Im' denotes `imaginary part'.
We denote by $p(f)$, the Stone generator of the strongly continuous one parameter unitary group $\{W(tf)|t\in\mathbb{R}\}$, that is  
$W(tf) = e^{-itp(f)}, t \in \mathbb{R}, f\in \mathcal{H}$.
The exponential domain $\mathcal{E}$, the dense subspace spanned by exponential vectors in $\Gamma_s(\mathcal{H})$,  is a core for $p(f)$ for all $f \in \mathcal{H}$. The space of all finite particle vectors, $\Gamma_s^0(\mathcal{H})$ is also a core for $p(f)$  for all $f$. 

Let us fix an orthonormal  basis  $\{e_j\}$ for $\CH$ and let

\begin{align}
  p_j&= 2^{-1/2}p(e_j),& q_j &= -2^{-1/2}p(ie_j)\\
a_j&=2^{-1/2}(q_j+ip_j) & a_j^{\dagger}&= 2^{-1/2}(q_j-ip_j)
\end{align}
for each $j \in \mathbb{N}$. Then we have the Lie brackets
\begin{align}
\label{eq:36}
[q_r,p_s] = i\delta_{rs}I, &\phantom{......} [a_r,a_s^{\dagger}]=\delta_{rs}
\end{align}
$\forall r,s \in \mathbb{N}$. Further $\{a_r, r\in \mathbb{N}\}$ and $\{a_r^{\dagger}, r\in \mathbb{N}\}$ commute among themselves. We call $p_j$ and $q_j$ as the $j$-th \emph{momentum} and \emph{position} operator,  $a_j$ and $a_j^{\dagger}$ as the $j$-th  \emph{annihilation} and \emph{creation} operator for all $j\in \mathbb{N}$. Let $z \in \mathcal{H}$ be such that  $z = \sum\limits_{j=1}^{n}\alpha_je_j$, where $\alpha_j= x_{j}+iy_j, x_j, y_j\in \mathbb{R},  \forall j$, then ~$$W(z)= e^{-i \sqrt{2} \sum\limits_{j=1}^n(x_jp_j-y_jq_j)}; ~~p(z) = \sqrt{2} \sum\limits_{j=1}^n(x_jp_j-y_jq_j).$$

If $T$ is any contraction on $\mathcal{H}$ then we define  $\Gamma(T)$ on the exponential vectors by 
\begin{equation}
\label{eq:73}
\Gamma(T)(e(f))=e(Tf),
\end{equation} which extends to a contraction on $\Gamma_{s}(\mathcal{H})$. Then $\Gamma(T)$ is called as the 
 \emph{second quantization} of $T$. It is possible to extend the definition of $\Gamma (T)$ via (\ref{eq:73}) even if $T~$ is a contraction mapping $\CH$ to a different Hilbert space $\mathcal{K}.$ 

If $U: \mathcal {H}\to \mathcal{K}$ is a unitary then we take $\Gamma_s(U):=\Gamma(U)$. Then $\Gamma_s(U)$ is  a unitary and we have
\begin{align}
\Gamma_{s}(U)^{-1} = \Gamma_s(U^{-1}); ~~
\Gamma_{s}(U)W(u)\Gamma_{s}(U)^{-1}=W(Uu).\label{eq:75}
\end{align}
Later we will extend the definition of $\Gamma _s(\cdot )$ for a large collection of operators called Shale operators.

  If $\CH =  \CH_1\oplus \CH_2$, then there is a unique unitary isomorphism between $\GH$ and $\Gamma_s(\CH_{1})\otimes \Gamma_s(\CH_{2})$ satisfying $e(f \oplus g) \mapsto e(f)\otimes e(g)$. Under this isomorphism we have $W(f\oplus g) = W(f)\otimes W(g)$. 

   Let $\mathcal{H}= \oplus_{n=1}^{\infty}\mathcal{H}_n$, where  $\mathcal{H}_n,  n= 1,2,3\dots$ is a sequence of Hilbert spaces. Consider the infinite tensor product $\otimes_{n=1}^{\infty}\allowbreak\Gamma_s(\mathcal{H}_{n})$, constructed using the stabilizing sequence $\{\Phi_n\}$, where $\Phi_n \in \Gamma_s(\mathcal{H}_n)$ is the vacuum vector for every $n.$ Then 
$\Gamma_s(\mathcal{H}) = \otimes_{n=1}^{\infty}\Gamma_s(\mathcal{H}_{n})$
under the natural isomorphism. In this identification, for $\oplus _{n=1}^{\infty}x_n \in \mathcal{H}, $ and contractions  $A_n \in \B{\mathcal{H}_{n}}, n\geq 1,$
\begin{align}
 e(\oplus_{n=1}^{\infty}x_n)&=  \otimes_{n=1}^{\infty}e(x_n) & :=\lim_{N\rightarrow\infty}\otimes_{j=1}^Ne(x_j)\otimes \otimes_{n=N}^\infty \Phi_n   \label{eq:97}\\
 W(\oplus_{n=1}^{\infty}x_n)& = \otimes_{n=1}^{\infty}W(x_n) & :=\slim_{N\rightarrow\infty}\otimes_{j=1}^NW(x_j)\otimes \otimes_{n=N}^\infty I \label{eq:8.2}\\
\Gamma (\oplus_{n=1}^{\infty}A_n) & = \otimes_{n=1}^{\infty}\Gamma(A_n)& :=\slim_{N\rightarrow\infty}\otimes_{j=1}^N\Gamma(A_j)\otimes \otimes_{n=N}^\infty I\label{eq:9.2} 
\end{align}


We refer to Section 20 of \cite{Par12} for more details and for proofs of all the above facts.

\subsection{\ensuremath{CCR} Algebra in Weyl form
}
In this Section we list some basic facts about symplectic spaces and quasifree states of CCR algebras. For more details  see \cite{Pet90, derezinkigerard, brattelirobinson2}.

 Let $H$ be a real linear space. A bilinear form $\sigma : H\times H \rightarrow \mathbb{R}$  is called a \emph{symplectic} form if $\sigma(f,g) = -\sigma(g,f)$, for every $f, g \in H$. The pair $(H, \sigma)$ is called a \emph{symplectic space}.  
  A symplectic form $\sigma$ on $H$ is called \emph{nondegenerate} if $\sigma(f,g)=0, \forall g \in H$ implies $f=0$.  
 A symplectic space $(H,\sigma)$ is called a \emph{standard (symplectic) space} if $H$ is a Hilbert space over $\mathbb{C}$ with respect to some inner product $\left\langle \cdot,\cdot \right\rangle$ and $\sigma(\cdot,\cdot ) = \im \left\langle \cdot,\cdot \right\rangle$.  
  It is called \emph{separable} if there exists in $H$ a countable family of vectors $\{f_k\}$ such that $\sigma(f,f_k)= 0 $ for all $k$ implies $f=0$. Note that  standard symplectic spaces are nondegenerate and separable.

    Let $(H,\sigma)$ be a nondegenerate symplectic space.   There exists a universal $C^*$-algebra generated by unitary elements $\{W(f): f \in H\}$ satisfying \begin{equation}
W(-f)=W(f)^{*}; ~~
W(f)W(g) = \exp(i\sigma(f,g))W(f+g), ~~\forall f, g \in H , \label{weyl2} 
\end{equation} which is unique up to isomorphism. This $C^*$-algebra is known as the algebra of canonical commutation relations in Weyl form and is denoted by $CCR(H,\sigma)$.  In view of (\ref{weyl2}), the linear hull of $\{W(f): f \in H\}$ is dense in $CCR(H,\sigma)$.
 

\subsubsection{Quasifree States on \ensuremath{CCR(H,\sigma)}}\label{sec:quasfree-states-ccrhsigma}
Let $\mathcal{A}$ be a $C^{*}$-algebra and let $\phi$ be a state on it. Let   $(H_{\phi},\Pi_{\phi},\Omega_{\phi})$ be the $GNS$ triple of $(\mathcal{A}, \phi ).$  A state $\phi$ on $\mathcal{A}$ is called \emph{primary} or a factor state if the von Neumann algebra $(\Pi_{\phi}(\mathcal{A)})''$ corresponding to the $GNS$-representation is a factor. It is called \emph{type I} if the von Neumann algebra $(\Pi_{\phi}(\mathcal{A)})''$ corresponding to the $GNS$-representation is a Type I factor. 

\begin{defn}\label{def:quasiequivalent} Two representations $\Pi_1$  and $\Pi_2$ of $\mathcal{A}$ are said to be \emph{quasiequivalent} if there exists a $*-$isomorphism of von Neumann algebras  $$\Theta: \Pi_{1}(\mathcal{A})'' \mapsto \Pi_{2}(\mathcal{A})''$$ satisfying $\Theta(\Pi_{1}(x)) =  \Pi_{2}(x)$ for all $x\in \mathcal{A}$.
  Two states $\phi$ and $\psi$ on $\mathcal{A}$ are said to be \emph{quasiequivalent} if  the corresponding $GNS$-representations are quasiequivalent. 
\end{defn}

 Let $(H,\sigma)$ be a standard symplectic space and let $\alpha: H\times H \rightarrow \mathbb{R}$ be a real inner product on $H$ with respect to which $H$ is complete.  Assume further that there exists a bounded, invertible real linear operator $A$ on $(H,\alpha)$ such that 
$\alpha(f,g) = \sigma(Af, g),$  for all $f, g \in H$ .
Then there exists a state (see \cite{Pet90})  $\phi_A$ on $CCR(H,\sigma)$ such that
\begin{equation}
\label{eq:2}
\phi_A(W(f)) = \exp \left( -\frac{1}{2} \alpha(f,f) \right) \phantom{.....} \forall f \in H.
\end{equation}
if and only if
\begin{equation}
\label{eq:1}
\sigma(f,g)^{2} \leq \alpha(f,f)\alpha(g,g), \phantom{.....} \forall f,g \in H.
\end{equation}
A state $\phi_A$ on $CCR(H,\sigma)$ determined in the form of (\ref{eq:2}) is called a \emph{quasifree state}.  
It is known that any quasfree state $\phi_A$ is primary (see  Proposition 1 and 2 of \cite{Hol71}). 
  Further in this case, 
%
%
%
\begin{equation}
 \label{eq:4}
A^{\tau}= -A ; \phantom{......} -A^2-I \geq 0 
\end{equation}
on $(H,\alpha(\cdot,\cdot))$.
 By $A^{\tau}$ we denote the transpose with respect to the inner product $\alpha$. Let us denote by $H_A$ the real Hilbert space $(H,\alpha(\cdot,\cdot))$. One major tool of the subject is the following theorem of Holevo. 
 

\begin{thm}[\cite{Hol71}] \label{quasi-equivalence}
   Two primary quasifree states $\phi_A$ and $\phi_B$ on a standard $CCR(H,\sigma)$ are quasi equivalent if and only if $A-B$ and $\sqrt{-A^2-I}-\sqrt{-B^2-I}$ are Hilbert-Schmidt operators  on $H_A$.
\end{thm}

\subsection{Shale maps and Bogoliubov transformations}\label{sec:sympl-autom}
Let $\mathcal{H}$ be a complex Hilbert space with inner product $\left\langle \cdot,\cdot \right\rangle$. Often we consider $\mathcal{H}$ as real Hilbert space also with $\left\langle \cdot,\cdot \right\rangle_{\mathbb{R}} = \re \left\langle \cdot,\cdot \right\rangle $. Let $H \subset \mathcal{H}$ be a real subspace such that $\mathcal{H}$ is the complexification $H+iH$ of $H$.
For any real linear operator $S$ on $\mathcal{H}$ , define operators $S_{ij}$ on $H$ such that 
$$
S(x+iy) = S_{11}x+iS_{21}x+S_{12}y+iS_{22}y.
$$
Define the operator $S_0$ on $H\oplus H$ by $S_0 = \begin{bmatrix}S_{11} & S_{12}\\S_{21}&S_{22}\end{bmatrix}$.
We identify $S$ with $S_0$ as a real linear operator and often switch between them freely. We do the same identification for any $S: \mathcal{H}\rightarrow \mathcal{K}$, for some other complex Hilbert space  $\mathcal{K}=K+iK$.
Note that if $S$ is a complex linear operator then $S_{11} = S_{22}$  and $S_{12} = -S_{21}$.


Let $J$ be the operator of multiplication by $-i$ on $\mathcal{H}$ considered as a real linear map. Then $$J_0 =
\begin{bmatrix}
  0 &I \\
 -I & 0
\end{bmatrix}.
$$
Clearly $J^*=J^{-1}=-J$ or equivalently $J_{0}^{\tau}=J_0^{-1}=-J_0$.

A real linear bijective map $L: \mathcal{H}\rightarrow\mathcal{K}$ is said to be a  \emph{symplectic transformation}  if it satisfies 
\begin{itemize}
\item [(i)] $L$ and $L^{-1}$ are continuous (bounded) ;
\item [(ii)] $\im \left\langle Lz, Lw \right\rangle = \im \left\langle z,w \right\rangle$ for all $z,w \in \mathcal{H}$.
\end{itemize}
Correspondingly, $L_0$ will also be called as symplectic transformation. If $\mathcal{K}=\CH$ then we call it a symplectic automorphism.

\begin{prop}[Section 22 in \cite{Par12}]
   $L: \mathcal{H}\to \mathcal{K}$ is symplectic if and only if $$L_0^{\tau}J_0L_0= J_0.$$  Here $J_0$ on left side is on $K\oplus K$ and that on the right side is on $ H\oplus H$.
\end{prop}

\begin{eg}\label{eg:fundamental-symplectic-opr}
  Let $A \in \B{H, K}$ be any invertible operator  on $H$ then the  operator $T: \CH \mapsto\mathcal{K}$ defined  by $T(u+iv) = Au +i (A^*)^{-1}v$ is a symplectic  transformation of $\CH$. 
\end{eg}
 We will have occasions to deal with the complexification of complex Hilbert spaces, considered as real Hilbert spaces under the inner product $\re\langle\cdot,\cdot\rangle$. 
We saw above that $(\mathcal{H}, \re\langle\cdot,\cdot\rangle ) $ has a canonical isomorphism to $H\oplus H$ as a real Hilbert space. 
Let $\hat{\mathcal{H}}$ denote its complexification, which is $\mathcal{H}\oplus \mathcal{H}$. If $A$ is a real linear operator on $\mathcal{H}$ then we denote by $\hat{A}$, the complexification of $A$, defined by $\hat{A}(z+iw) = Az+iAw$. It is easy to verify, for a a real linear operator $S$ on $\mathcal{H}$ with $S_0 =
  \begin{bmatrix}
    S_{11}&S_{12}\\
    S_{21}&S_{22}
  \end{bmatrix}$, that $\hat{S} =
\begin{bmatrix}
  \hat{S_{11}}&\hat{S_{12}}\\
 \hat{S_{21}} & \hat{S_{22}}
\end{bmatrix}
$ on $\mathcal{H}\oplus \mathcal{H}$. 
In particular we have
$\hat{J}  =
   \begin{bmatrix}
     0& I\\
    -I&0
   \end{bmatrix}
$ on $\mathcal{H}\oplus \mathcal{H}$.

By the same analysis as in  Proposition 22.1 in \cite{Par12} and by applying polar decomposition, we arrive at the following Proposition.

\begin{prop}\label{prop:factoring-symplectic-transf}
 Let $\mathcal{H}, \mathcal{K}$ be complex Hilbert spaces and let $S:\CH\rightarrow\mathcal{K}$ be symplectic. Then it admits a decomposition:
\begin{equation}
\label{eq:57}
S = UTV
\end{equation} where 
$U: \CH\rightarrow\mathcal{K}$ and $
V: \CH\rightarrow\CH $ are unitaries and 
$T: \CH\rightarrow\CH $ has the form $$ T(u+iv) = Au+iA^{-1}v,$$
 where $A \in \B{H}$ is a positive and invertible operator. 
\end{prop}

Suppose $\mathcal{H}$ is a complex Hilbert space. We look at automorphisms of the collection of Weyl operators on $\GH .$ Shale's theorem \cite{Sha62} characterizes the set of all bounded, invertible real linear maps $L: \mathcal{H}\to \mathcal {H}$ which  admit some unitary operator $Z$ (unique up to a scalar multiple of modulus one) on the Fock space $\GH $ such that $ZW(u)Z^*= W(Lu), u\in \mathcal{H}.$  A lucid presentation of this can be seen in Section 22 of \cite{Par12}. 
We state the theorem here with some obvious modifications required for our work and sketch a proof of the same based on the construction given in \cite{BhSr05}.  

Let $\mathcal{H}, \mathcal{K}$ be two Hilbert spaces, define the set of \emph{Shale Operators} from $\CH$ to $\mathcal{K}$, $\mathscr{S}(\mathcal{H},\mathcal{K})$ by 
\[
\mathscr{S}(\mathcal{H},\mathcal{K}) = \{L \in \mathscr{B}_{\mathbb{R}}(\mathcal{H},\mathcal{K}): L \textnormal{ is symplectic and } L^{\tau}L - I \textnormal{ is Hilbert-Schmidt}.\}
\]
We denote $\mathscr{S}(\mathcal{H}):=\mathscr{S}(\mathcal{H},\mathcal{H})$.
Notice from Example \ref{eg:fundamental-symplectic-opr} that there exists symplectic transformations  $T$ such that $T \notin \mathscr{S}(\mathcal{H},\mathcal{K})$.


\begin{thm}\label{thm:generalized-shale-unitary}
\begin{enumerate}
\item\label{item:21} 
 Let $L \in \mathscr{B}_{\mathbb{R}}(\mathcal{H},\mathcal{K})$ be a symplectic operator.  
 Then there exists a unitary operator $\Gamma_s(L): \GH\rightarrow \Gamma_s(\mathcal{K})$ such that 
 \begin{equation}
\label{eq:56}
\Gamma_s(L)W(u)\Gamma_s(L)^{*} = W(Lu), \forall u \in \CH, 
\end{equation}
if and only if $L \in \SH{\CH, \mathcal{K}}.$  In such a case, $\Gamma_s(L)$ is determined uniquely up to a scalar of modulus unity.
\item \label{item:21.1} A unitary $\Gamma_s(L)$ satisfying (\ref{eq:56}) can be chosen such that it satisfies \begin{equation}
    \label{eq:60}
\left\langle \Gamma_s(L)\Phi_{\CH}, \Phi_{\mathcal{K}} \right\rangle \in \mathbb{R}^{+},
\end{equation}
where $\Phi_{\mathcal{H}}$ and $ \Phi_{\mathcal{K}}$ are vacuum vectors in $\GH$ and $\Gamma_s(\mathcal{K})$ respectively, this choice makes $\Gamma_s(L)$ unique. In this case, 
\begin{equation}\label{eq:71}
    \Gamma_s(L^{-1})= \Gamma_s(L)^{*}
\end{equation}
\item \label{item:22} Let $ \mathcal{H}_{1},\mathcal{H}_{2},\mathcal{H}_{3}$ be three Hilbert spaces and $L_1 \in \SH{\mathcal{H}_1,\mathcal{H}_{2}}, L_2\in \SH{\mathcal{H}_2,\mathcal{H}_3}$. Then 
\begin{equation}
\label{eq:77}
\Gamma_s(L_2L_1) = \sigma(L_2, L_1)\Gamma_s(L_2)\Gamma_s(L_1),
\end{equation}
where $\sigma(L_2, L_1) \in \mathbb{C}, \abs{\sigma(L_2, L_1)} = 1$. 
\end{enumerate}
\end{thm}

\begin{proof}\ref{item:21}. 
Assume that $L \in \SH{\CH, \mathcal{K}}.$ We will prove the existence of $\Gamma_s(L)$ based on the construction in \cite{BhSr05}. By Proposition \ref{prop:factoring-symplectic-transf} there exist unitaries $U: \CH \rightarrow \mathcal{K}$, $V: \CH \rightarrow \CH$ such that $L = UTV$ where $T$ is a symplectic automorphism of $H$ such that
  \begin{equation*}
T(u+iv) = Au+iA^{-1}v    
  \end{equation*}
where $A \in \B{H}$ is  positive and invertible.  It can be seen from the proof of Proposition \ref{prop:factoring-symplectic-transf} that
\begin{equation*}
  L_0 = U_0
  \begin{bmatrix}
    A &0\\
   0& A^{-1}
  \end{bmatrix} V_0
\end{equation*}
 for some orthogonal transformations $U_0 \in \B{H, K}$ and $V_0\in \B{H}$. Now it can be seen that \[L_0^{T}L_0 = V_0^{-1}
   \begin{bmatrix}
     A^2&0\\
     0 & A^{-2}
   \end{bmatrix}V_0.
\]
Therefore $L_0^{T}L_0-I = V_0^{-1}\left(
\begin{bmatrix}
  A^2&0\\
 0 & A^{-2}
\end{bmatrix} -
\begin{bmatrix}
  I &0\\0&I
\end{bmatrix}
\right)V_0
$. Hence we get that $A^{2} - I$ is Hilbert-Schmidt and since A is positive, Theorem 2.1 of \cite{BhSr05} applies. Thus there exists $\Gamma_s(T)$ such that 
\begin{align}
\label{eq:58}
\Gamma_s(T)W(u)\Gamma_s(T)^{*}& = W(Tu), \forall u \in \CH,\\ \label{eq:59}
\left\langle \Gamma_s(T)\Phi_{\CH},\Phi_{\mathcal{K}} \right\rangle& \in \mathbb{R}^+.
\end{align}
Define 
\begin{equation}
\label{eq:72}
\Gamma_s(L):= \Gamma_s(U)\Gamma_s(T)\Gamma_s(V),
\end{equation} where $\Gamma_s(U)$ and $\Gamma_s(V)$ are the second quantization associated with the unitary $U$ and $V$. A direct computation shows that $\Gamma_s(L)$ satisfies the (\ref{eq:56})(because of properties of $\Gamma_s(U),\Gamma_s(V)$ and equation \ref{eq:59}) and (\ref{eq:60}) (because second quantizations $\Gamma_s(U_j)$ acts as identity on vacuum vector). We refer to Theorem 22.11 in \cite{Par12} for the necessity part.

\ref{item:21.1}. Equation (\ref{eq:60}) is automatically satisfied in our construction in (1) above because of (\ref{eq:59}). To see the uniqueness, let $\Gamma_s^1(L)$ and $\Gamma_s^2(L)$ satisfy  (\ref{eq:56}) and (\ref{eq:60}). Therefore we get $\Gamma_s^2(L)^*\Gamma_s^1(L)W(u) =W(u)\Gamma_s^2(L)^*\Gamma_s^1(L), \forall u\in \CH $. Therefore by irreducibility of Weyl operators (Proposition 20.9 in \cite{Par12}), $\Gamma_s^2(L)^*\Gamma_s^1(L) = cI$ for some complex scalar of unit modulus. But now by (\ref{eq:60}) we get $\Gamma_s^2(L) = \Gamma_s^1(L)$.

To prove  (\ref{eq:71}), note that $\left\langle \Gamma_s(L)^{*} \Phi_{\mathcal{K}},  \Phi_{\CH} \right\rangle = \left\langle \Gamma_s(L)\Phi_{\CH}, \Phi_{\mathcal{K}} \right\rangle \in \mathbb{R}^{+}$ therefore if we show that $\allowbreak\Gamma_s(L)^{*} W(u)\Gamma_s(L) = W(L^{-1}u)$ then by the uniqueness of $\Gamma_s(L^{-1})$  we get  (\ref{eq:71}). Recall from Theorem 2.1 of \cite{BhSr05} that $\Gamma_s(T^{-1}) = \Gamma_s(T)^{*}$ and for second quantization (\ref{eq:73}) unitary we have, $\Gamma_s(U_j^{*}) = \Gamma_s(U_j)^{*}$. Further by (\ref{eq:72}), and (\ref{eq:75}) we have
\begin{align*}
\Gamma_s(L)^{*} W(u)\Gamma_s(L)&=\Gamma_s(U_2)^{*}\Gamma_s(T)^{*}\Gamma_s(U_1)^{*}W(u) \Gamma_s(U_1)\Gamma_s(T)\Gamma_s(U_2)\\
& = W(U_2^{*}T^{-1}U_1^{*}u)\\
& = W(L^{-1}u).
\end{align*}
This completes the proof of \ref{item:21.1}.

\ref{item:22}. Follows immediately from \ref{item:21}.
\end{proof}
\begin{rmks}\begin{enumerate}

    \item When $\CH$ is finite dimensional, it is known that there exists a choice of $\Gamma_s(L)$ such that the multiplier $\sigma(L_1, L_2) = \pm 1, \forall L_1, L_2 \in Sp_{2n}(\mathbb{R})$, where $Sp_{2n}(\mathbb{R})\subseteq M_{2n}(\mathbb{R})$ is the subgroup of all $2n\times 2n$ symplectic matrices. This is called the metaplectic representation of the symplectic group. An elementary and self-contained presentation can be found in Chapter 4 of  \cite{Folland89}, Theorem 4.37 there is of particular interest in this regard. In the infinite dimensional case, \cite{MATSUI200467} and \cite{Tveritinov2004} are of interest.
    
    \item The map $W(u)\mapsto W(Lu)$ is known as the \emph{Bogoliubov transformation} of the CCR algebra, induced by $L$. Whenever we write $\Gamma_s(L)$, we mean the unique unitary operator satisfying (\ref{eq:60}). It is called  the \emph{Shale unitary} corresponding to an $L \in \mathscr{S}(\mathcal{H},\mathcal{K})$. It is to be noted that if $L$ is a  non-unitary contraction then $\Gamma (L)$ defined by (\ref{eq:73}) is not a unitary and hence in such a case $\Gamma _s (L)\neq \Gamma (L).$
\end{enumerate}

\end{rmks}
 
 \subsection{Williamson's normal form and symplectic spectrum}
 
The following Theorem  \cite{BhJo18} extends well-known Williamson's normal form to bounded operators on infinite dimensional real Hilbert spaces.

\begin{thm} (Williamson's Normal Form) \label{sec:will-norm-form}
  Let $H$ be a real Hilbert space and let $S_{0}$ be a strictly positive invertible operator on $H\oplus H$ then there exists
a real Hilbert space $K$, a positive invertible operator $\tilde{P}$ on $K$ and a symplectic transformation $ M_0\colon H\oplus H \rightarrow K \oplus K$
 such that   \begin{equation}
    \label{eq:will-norm-form}
    S_0 = M_0^{\tau} \begin{bmatrix}
             \tilde{P} & 0  \\
              0 & \tilde{P}
             \end{bmatrix} M_0
  \end{equation}
The decomposition is unique in the sense that if $P_1$ is any strictly positive invertible operator on a Hilbert space $\tilde{H}$ and $\tilde{M}\colon H\oplus H \rightarrow \tilde{H} \oplus \tilde{H}$ is a symplectic transformation such that
 \begin{equation}\label{eq:uniqueness}
    S_0 = \tilde{M}^{T} \begin{bmatrix}
              P_1 & 0  \\
              0 & P_1
             \end{bmatrix} \tilde{M},
  \end{equation}
then
$P$ and $P_1$ are orthogonally equivalent.
\end{thm}

 For the work in this article we need a version of Williamson's Normal form which produces a diagonalization within the original Hilbert space $\CH$ itself. We will identify $\CH$ and $\mathcal{K}$ accordingly for this purpose. We obtain this goal in Corollary \ref{cor:will-norm-form-complex-case}. The scheme is as described below.

 Since $M_0$ in (\ref{eq:will-norm-form}) is an invertible operator between $H\oplus H$ and $K \oplus K$; $H$ and $K$ have same dimension. Let $\{h_j\}$ and $\{k_j\}$ be a basis for $H$ and $K$ respectively. Since $\mathcal{H} = H+iH$ and $\mathcal{K} = K+iK$, there exists a unitary, $U:\mathcal{K} \rightarrow \CH$ such that $U(h_j) = k_j$. 
 If we write $\tilde{\mathscr{P}}_0 = \begin{bsmallmatrix}
              \tilde{P} & 0  \\
              0 & \tilde{P}
             \end{bsmallmatrix}$, then $U\tilde{\mathscr{P}}U^*:= \mathscr{P}$ is a complex linear, positive and invertible operator on $\CH$.
 Recall the relationship between the operator $\mathscr{P}$ defined on $\mathcal{H}$ and $\mathscr{P}_0$ on $H$ described at the beginning of \ref{sec:sympl-autom}. Notice now, for the corresponding real linear map $\mathscr{P}_{0}$, we have $\mathscr{P}_{0} = \begin{bsmallmatrix}
              P & 0  \\
              0 & P
             \end{bsmallmatrix}$ for some positive invertible operator $P$ on $H$.  Now we have the following Corollary. 
              
  \begin{cor}\label{cor:will-norm-form-complex-case}
    Let $\CH= H+iH$ and $S$ be a real linear positive, invertible operator on a complex Hilbert space $\mathcal{H}$. Then there exists a  complex linear positive invertible operator $\mathscr{P}$ (unique up to a unitary conjugation)  and a symplectic automorphism $L$ on $\mathcal{H}$ such that 
\begin{equation}
\label{eq:25}
S = L^{\tau}\mathscr{P}L.
\end{equation}
Further, $\mathscr{P}$ has the property that $\mathscr{P}_0 =
\begin{bsmallmatrix}
  P & 0 \\
 0 & P
\end{bsmallmatrix}
$ and thus we also have \begin{equation}
    \label{eq:will-norm-form-sames-space}
    S_0 = L_0^{\tau} \begin{bmatrix}
             P & 0  \\
              0 & P
             \end{bmatrix} L_0.
  \end{equation}
  \end{cor}
  \begin{proof}
  We have $S_0 = M_0^{\tau}\tilde{\mathscr{P}}_0M_0$ by Theorem \ref{sec:will-norm-form}. Therefore $S= M^{\tau}U^*U\tilde{\mathscr{P}}U^*UM$. Take $L= UM$. It should be noticed that the transpose of $U$ considered as a real linear operator is same as the $U^*$.
  \end{proof}
  \begin{rmk}\label{rmk:symplectic-spectrum}
Under the situation of Corollary \ref{cor:will-norm-form-complex-case}, in view of the uniqueness of the decomposition, the spectrum of $\mathscr{P}$, can be defined as the {\em symplectic spectrum\/}  of the positive invertible operator $S$.
\end{rmk}
\section{Quantum Gaussian States}

In this Section we present the basics of infinite mode quantum Gaussian states. Perhaps most of the results in this Section are known. Since notation as well as the set up differ in different sources (some of them consider only complex linear operators),  for the convenience of the reader we present some of the steps.

Recall that a \textit{state (or density operator)} $\rho$ on a Hilbert space $\mathcal{H}$ is a bounded positive operator of unit trace. Such a density operator uniquely determines a state  on the $C^{*}$- algebra $\B{\CH}$ by $Y\mapsto \tr \rho Y$, $Y \in \B{\CH}$.

\begin{defn}\label{def:QFT}
  Let $\rho \in \B{\Gamma_s(\mathcal{H})}$  be a density operator. Then the complex valued function $\hat{\rho}$ on $\mathcal{H}$ defined by 
\begin{equation}
\label{eq:10}
\hat{\rho}(z) = \tr \rho W(z), \phantom{......} z \in \mathcal{H}
\end{equation}
 is called the \emph{quantum characteristic function}(or \emph{quantum Fourier transform}) of $\rho$.
\end{defn}
 As the von Neumann algebra generated by Weyl unitaries is the algebra of all bounded operators it follows that the map  $\rho \rightarrow \hat{\rho}$ is injective. Here is an analogue of well-known positive definiteness of  characteristic functions of probability measures.  If $\rho$ is any density operator, then the kernel $K_{\rho}$ on $\mathcal{H}\times \mathcal{H}$ defined by $K_{\rho}(z,w) = e^{i \im{\left\langle z,w \right\rangle}}\hat{\rho}(w-z)$ is positive definite, i.e $\sum_{j,k = 1}^{n}\overline{c_{j}}c_{k}K_{\rho}(z_{j},z_{k})\geq 0$, for all $(c_1, c_2 \cdots c_n) \in \mathbb{C}^n, n\in \mathbb{N}$. The proof is a direct computation making use of commutation relations of Weyl unitaries (see  \cite{Par10}).

If $\rho$ is a state so is any unitary conjugation of it. 
 By using Theorem \ref{thm:generalized-shale-unitary}, proof of the following proposition follows in the same way as that of Proposition 2.5 in \cite{Par10}.
\begin{prop}\label{prop:qcf-of-weyl-conjugation}
 If $\rho$ is a state on $\Gamma_s(\mathcal{K})$ and $L \in  \mathscr{S}(\CH, \mathcal{K})$ 
 then
\[\{\Gamma_s(L)^{*}\rho\Gamma_s(L)\}^{\wedge}(\beta) = \hat{\rho}(L\beta).\]
Further, for every $\alpha \in \CH$, 
\[\{W(\alpha)\rho W(\alpha)^{-1}\}^{\wedge}(\beta) = \hat{\rho}(\beta) e^{2i \im \left\langle \alpha, \beta \right\rangle} .\] 
\end{prop}

\begin{defn}\label{defn:quant-gauss-stat-QFT}
Let $\rho \in \B{\GH}$ be a state, $\rho$ is said to be Gaussian if there exists $w\in \CH$ and a symmetric, invertible $S \in \mathscr{B}_{\mathbb{R}}(\CH)$ 
such that 
\begin{equation}
\label{eq:47}
\hat{\rho}(z) = \exp{-i \re \left\langle w, z \right\rangle - \frac{1}{2} \re  \left\langle z,
Sz \right\rangle}, \forall z \in  \CH. 
  \end{equation}
In such a case we write  $\rho = \rho_g(w,S)$.
\end{defn}
Note that this definition determines a real linear functional $z\mapsto \re \left\langle w, z \right\rangle $ and a bounded  quadratic form $z \mapsto \re  \left\langle z, Sz \right\rangle$ on the real Hilbert space $\CH$. Hence $w$ and $S$ are uniquely determined by the definition.                                                                                                                                                                                                                                                                                                  We call $w$ the \emph{mean vector} and $S$ the \emph{covariance operator} associated with $\rho$. Suppose $\CH = H+iH$, where $H$ is a  real subspace and let $w= \sqrt{2}( l-im)$, then we call $l$ and $m$ as \emph{mean momentum vector} and \emph{mean position vector} respectively. 

 Let $\mathbf{G}(\CH)$ denote the set of all Gaussian states on $\GH$ and $\mathscr{K}_{\mathbf{G}}(\CH)$ denote the set of all Gaussian covariance operators on $\CH$.
We will characterize the elements of  $\mathscr{K}_{\mathbf{G}}(\CH)$ in Theorem \ref{thm:main-thm}.

\begin{egs}\label{egs:Gaussian-states}
 \begin{enumerate}
\item \label{eg:coherent-state} For $f \in \CH$, \phantom{|}consider  \phantom{|}the normalized exponential vector  $\psi(f):= e^{-\frac{1}{2}\|f\|^2}e(f)$. The pure state $\ketbra{\psi(f)}{\psi(f)}$ is called  as the \emph{coherent state}. The same arguments as that  of Proposition 2.9 in \cite{Par10}, prove that any coherent state is a pure Gaussian state on $\GH$ and
\begin{equation}
\label{eq:62}
(\ketbra{\psi(f)}{\psi(f)})^{\wedge}(z) = \exp{-2i \im \left\langle z,f \right\rangle -\frac{1}{2}\|z\|^2}
\end{equation}
In particular
$\ketbra{e(0)}{e(0)} = \rho_g(0,I)$.
\item\label{shale-Gaussian}   Let $L$ be a symplectic automorphism on $\CH$ such that $L^{\tau}L-I$ is Hilbert-Schmidt. Define $\psi_L = \Gamma_s(L)^{*}\ket{e(0)}$. Then 
\begin{align*}
(\ketbra{\psi_L}{\psi_L})^{\wedge}(z) 
                               &= \tr  (\ketbra{\psi_L}{\psi_L}W(z))\\
                               & = \left\langle \psi_L, W(z)\psi_L \right\rangle\\
                               &= \left\langle e(0), \Gamma_s(L)W(z)\Gamma_s(L)^{*}e(0) \right\rangle\\
                               &=  \left\langle e(0), W(Lz)e(0) \right\rangle\\
                               &= e^{-\frac{1}{2}\left\langle z, L^{\tau}Lz \right\rangle}.
\end{align*}
Therefore, $\ketbra{\psi_L}{\psi_L} = \rho_g(0, L^{\tau}L)$.
\item\label{eg:fundamental-mixed-gaussian}
Consider $\Gamma_s(\mathbb{C})= L^2(\mathbb{R})$, under this isomorphism there is an orthonormal basis, $\{\psi_n\}$ of $L^2(\mathbb{R})$, in which $e(z) = \sum\limits_{n=0}^\infty \frac{z^n}{\sqrt{n!}}\psi_n$. Let us notate the annihilation and creation operators as $a$ and $a^{\dagger}$ respectively. Then the number operator,  $a^\dagger a$ satisfies  $a^\dagger a\psi_n = n\psi_n, \forall n \in \mathbb{N}$ and \[\tr e^{-sa^{\dagger}a} = (1-e^{-s})^{-1},  \phantom{...} s>0.\]
Therefore the states 
\begin{equation}
    \rho_s = (1-e^{-s})e^{-sa^{\dagger}a}, \phantom{...} s>0
\end{equation}
are well defined. By Proposition 2.12 in \cite{Par10} we have
\begin{equation}\label{eq:QFT-fundamental-mixed-gaussian}
    \hat{\rho}_s(z) = \exp{-\frac{1}{2}(\coth{\frac{s}{2}})\abs{z}^2}.
\end{equation}
 Therefore $\rho_s \in \mathbf{G}(\mathbb{C})$. Since the spectrum of $a^{\dagger}a$ is $\{0,1,2,\dots\}$, it is not a pure state. It may be noted that $e^{-sa^\dagger a} = \Gamma(e^{-s})$, the second quantization (\ref{eq:73}) of the contraction $e^{-s}I$.
\end{enumerate}
\end{egs}

    All the statements in the following Proposition follow by direct  computations. 

\begin{prop}\label{prop:weyl-conjugation-Gaussian-state}
\begin{itemize}
  \item [(i)] Let $\alpha \in \CH$. Then \[W(\alpha)\rho_g(w, S)W(\alpha)^{-1} = \rho_g(w-2i\alpha, S).\]
In particular, \[ W(\frac{-i}{2}w)\rho_g(w, S)W(\frac{-i}{2}w)^{-1} = \rho_g(0,S).\]
\item [(ii)]
   Let $\rho_1 = \rho_g(w, S_1)$ and $\rho_2= \rho_g(w_2,S_2)$ be Gaussian states on $\Gamma_s(\CH_{1})$ and $\Gamma_s(\CH_{2})$ respectively. Then $\rho_{1}\otimes \rho_{2} = \rho_g(w_1\oplus w_2, S_1\oplus S_2)$.
\item [(iii)]  If $\rho = \rho_g(w, S)$ on $\Gamma_s(\mathcal{K})$ and $L \in  \mathscr{S}(\CH, \mathcal{K})$ then
\[\Gamma_s(L)^{*}\rho\Gamma_s(L) = \rho_g(L^{\tau}w, L^{\tau}SL).\]
\end{itemize}
\end{prop}
 The above proposition asserts that the Weyl conjugation displaces  the mean (hence the name displacement operator in physics literature) and the Shale unitary transforms the covariance operator. 

\subsection{Necessary  conditions on the covariance operator}

We would like to characterize covariance operators of quantum Gaussian states. In the following lemma we see that  if $S$ is one such operator on $\CH$  then  $\hat{S}-i \hat{J}\geq 0$ on $\hat{\CH}$. It is to be noted that here $\CH$ is a complex Hilbert space, but we are considering it as a real Hilbert space (under the real part of the given complex inner product) and then complexifying it  to get $\hat{\CH}.$ 

\begin{lem}\label{first lemma}
  Let $S$ be a real linear symmetric and invertible  operator on $\mathcal{H}$. Suppose the function $f:\mathcal{H}\rightarrow \mathbb{R}$ defined by $f(z)= e^{-\frac{1}{2}\re\left\langle z,Sz \right\rangle}$ is the quantum characteristic function of a density operator $\rho$ i.e. $S \in \mathscr{K}_{\mathbf{G}}(\CH)$  then  on $\hat{\CH}$, 
\begin{equation}
\label{eq:17}
\hat{S}-i \hat{J}\geq 0.
\end{equation} 
\end{lem}

\begin{proof}
 We have already noticed that on $f$ being a quantum characteristic function, the kernel 
\begin{equation}
\label{eq:66}
K_{\rho}(\alpha,\beta) = e^{i\im \left\langle \alpha,\beta \right\rangle}f(\beta-\alpha), \phantom{....}\alpha,\beta \in \CH.
\end{equation}
is  positive definite on $\mathcal{H}$. If $\alpha = x +iy, \beta = u+iv$ where $x,y,u,v \in H$ then
$\im \left\langle \alpha,\beta \right\rangle = \left\langle\begin{psmallmatrix}x\\y
\end{psmallmatrix}, J_0\begin{psmallmatrix}u\\v
\end{psmallmatrix}
\right\rangle$ on $H \oplus H$ (c.f.~\ref{sec:sympl-autom})
we can rewrite the definition of $K_{\rho}$ as 
\begin{equation}
\label{eq:8.1}
 K_{\rho}(\alpha,\beta)= \exp\left\{i\left\langle\begin{psmallmatrix}x\\y
\end{psmallmatrix}, J_0\begin{psmallmatrix}u\\v
\end{psmallmatrix}
\right\rangle - 
\left\langle \begin{psmallmatrix}u-x\\  v-y \end{psmallmatrix},
\frac{1}{2}S_0\begin{psmallmatrix}
  u-x\\
  v-y
\end{psmallmatrix} \right\rangle \right\}
\end{equation}
Now positive definiteness of $K_{\rho}$ in $\mathcal{H}$ reduces to that of $L$ in $H\oplus H$ where 
\begin{equation}
\label{eq:9.1}
 L \left((x,y),(u,v) \right)= \exp\left\{i\left\langle\begin{psmallmatrix}x\\y
\end{psmallmatrix}, J_0\begin{psmallmatrix}u\\v
\end{psmallmatrix}
\right\rangle - 
\left\langle \begin{psmallmatrix}u-x\\  v-y \end{psmallmatrix},
\frac{1}{2}S_0\begin{psmallmatrix}
  u-x\\
  v-y
\end{psmallmatrix} \right\rangle \right\}
\end{equation}
This is equivalent to the positive definiteness of $$ L_t \left((x,y),(u,v) \right) =  L \left(\sqrt{t}(x,y),\sqrt{t}(u,v) \right)$$
for all $t\geq0$. But $\{L_t\}$ is a one parameter multiplicative semigroup of kernels on $H\oplus H$. By elementary properties of positive definite kernels as described in Section 1 of \cite{PaSc72}, positive definiteness of $L_t, t\geq 0$ is equivalent to the conditional positive definiteness of  
$$ N \left((x,y),(u,v) \right)= i\left\langle\begin{psmallmatrix}x\\y
\end{psmallmatrix}, J_0\begin{psmallmatrix}u\\v
\end{psmallmatrix}
\right\rangle - 
\left\langle \begin{psmallmatrix}u-x\\  v-y \end{psmallmatrix},
\frac{1}{2}S_0\begin{psmallmatrix}
  u-x\\
  v-y
\end{psmallmatrix} \right\rangle $$
or equivalently (by the Lemma 1.7 in \cite{PaSc72}), the positive definiteness of 
\begin{align*}
 &N \left((x,y),(u,v) \right)- N \left((x,y),(0,0) \right)- N \left((0,0),(u,v) \right)- N \left((0,0),(0,0) \right)\\
&=i\left\langle\begin{psmallmatrix}x\\y
\end{psmallmatrix}, J_0\begin{psmallmatrix}u\\v
\end{psmallmatrix}
\right\rangle - 
\left\langle \begin{psmallmatrix}u-x\\  v-y \end{psmallmatrix},
\frac{1}{2}S_0\begin{psmallmatrix}
  u-x\\
  v-y
\end{psmallmatrix} \right\rangle 
 +\left\langle\begin{psmallmatrix}x \\ y\end{psmallmatrix},
                  \frac{1}{2}S_0\begin{psmallmatrix}
                  x\\
                  y
                 \end{psmallmatrix} \right\rangle +
                               \left\langle \begin{psmallmatrix}u \\ v \end{psmallmatrix},
                               \frac{1}{2}S_0\begin{psmallmatrix}
                                      u\\
                                      v
                                \end{psmallmatrix} \right\rangle\\
&=i\left\langle\begin{psmallmatrix}x\\y
\end{psmallmatrix}, J_0\begin{psmallmatrix}u\\v
\end{psmallmatrix}
\right\rangle + 
\left\langle \begin{psmallmatrix}x\\  y \end{psmallmatrix},
\frac{1}{2}S_0\begin{psmallmatrix}
  u\\
  v
\end{psmallmatrix} \right\rangle 
 +\left\langle\begin{psmallmatrix}u \\ v\end{psmallmatrix},
                  \frac{1}{2}S_0\begin{psmallmatrix}
                  x\\
                  y
                 \end{psmallmatrix}\right \rangle \\
&=i\left\langle\begin{psmallmatrix}x\\y
\end{psmallmatrix}, J_0\begin{psmallmatrix}u\\v
\end{psmallmatrix}
\right\rangle + 
\left\langle \begin{psmallmatrix}x\\  y \end{psmallmatrix},
\frac{1}{2}S_0\begin{psmallmatrix}
  u\\
  v
\end{psmallmatrix} \right\rangle 
 +\left\langle\begin{psmallmatrix}x \\ y\end{psmallmatrix},
                  \frac{1}{2}S_0\begin{psmallmatrix}
                  u\\
                  v
                 \end{psmallmatrix}\right \rangle  \numberthis \label{eq:10.1}\\
&=i\left\langle\begin{psmallmatrix}x\\y
\end{psmallmatrix}, J_0\begin{psmallmatrix}u\\v
\end{psmallmatrix}
\right\rangle + 
\left\langle \begin{psmallmatrix}x\\  y \end{psmallmatrix},
S_0\begin{psmallmatrix}
  u\\
  v
\end{psmallmatrix} \right\rangle \\
&=i\left\langle\begin{psmallmatrix}u\\v
\end{psmallmatrix}, -J_0\begin{psmallmatrix}x\\y
\end{psmallmatrix}
\right\rangle + 
\left\langle \begin{psmallmatrix}u\\  v \end{psmallmatrix},
S_0\begin{psmallmatrix}
  x\\
  y
\end{psmallmatrix} \right\rangle \numberthis \label{eq:10.3}\\
      \end{align*}
where (\ref{eq:10.1}) follows because the real inner-product is symmetric and $S_0$ is symmetric, 
and (\ref{eq:10.3}) for a similar reason. But  $H\oplus H \subset \mathcal{\hat{H}} = (H\oplus H)+i( H\oplus H)$, the positive definiteness of (\ref{eq:10.3}) lifts to the positive definiteness of

\begin{equation}
\label{eq:16} 
M(w,z) := \left\langle w, \left\{\hat{S}-i\hat{J}\right\}z \right\rangle =
\left\langle\begin{psmallmatrix}u\\v
\end{psmallmatrix}, -i\hat{J}_0\begin{psmallmatrix}x\\y
\end{psmallmatrix}
\right\rangle + 
\left\langle \begin{psmallmatrix}u\\  v \end{psmallmatrix},
\hat{S}_0\begin{psmallmatrix}
  x\\
  y
\end{psmallmatrix} \right\rangle 
\end{equation}
where $M$ is a kernel defined (as above) in  $\mathcal{H} \subset \mathcal{\hat{H}}$.
 Now, since   $\hat{S}-i\hat{J} $ is self-adjoint, 
 the positive definiteness of $M$ in (\ref{eq:16})  is equivalent to (\ref{eq:17}).
\end{proof}



\begin{lem}\label{lem:positivity-condition-consequences}
Let $S$ be a real linear, invertible  operator on $\mathcal{H}$ and $\hat{S}-i \hat{J}\geq 0$ on $\hat{\CH}$. Then 
\begin{enumerate}
\item \label{item:10} $S \geq 0$.
\item\label{item:9} If  $S = L^{\tau}\mathscr{P}L$ is the Williamson's normal form associated with $S$ (as in Corollary \ref{cor:will-norm-form-complex-case}), then $\mathscr{P} -I \geq 0$ on $\CH$.
\item \label{item:8} There exists a primary quasifree state $\phi$ on $CCR(\mathcal{H},\sigma)$ such that 
\begin{equation}\label{eq:52}
\phi(W(z)) = e^{-\frac{1}{2}\re\left\langle z,Sz \right\rangle}.
\end{equation}
   Further,  $\phi = \phi_A$, where $A=-JS$ (the notation $\phi_A$ is as in \ref{sec:quasfree-states-ccrhsigma}).
\end{enumerate} 
\end{lem}

\begin{proof} \ref{item:10}.
 Note $\hat{S}-i\hat{J} \geq 0$ implies $\hat{S}$ is symmetric, hence  $S$ is also symmetric. 
 Let $z, w \in H \oplus H$, so that   $z+iw \in \hat{\mathcal{H}}$. Then
\begin{align*}
 0 &\leq \left\langle z+iw, (\hat{S}-i\hat{J})z+iw \right\rangle\\ &= \re \left\langle z, Sz \right\rangle + i \re \left\langle z,Sw \right\rangle -i \re \left\langle w, Sz \right\rangle + \re \left\langle w, Sw  \right\rangle\\ 
                                                        & \phantom{=} - i \re \left\langle z, Jz \right\rangle + \re \left\langle z, Jw \right\rangle- \re \left\langle w, Jz \right\rangle - i\re \left\langle w, Jw \right\rangle\\
                                                        & =\re \left\langle z, Sz \right\rangle + \re \left\langle w, Sw  \right\rangle+2 \re \left\langle z, Jw \right\rangle \numberthis \label{eq:5.1}   \end{align*}
where we used the facts 
$\re \left\langle z, Jz \right\rangle = 0$ for all $z$. Letting $z=w$ in the above computation, then we get $S \geq 0$.

\ref{item:9}.
  Let  $\mathscr{P}_0 =\begin{bsmallmatrix}
              P & 0  \\
              0 & P
             \end{bsmallmatrix}$. Then $\hat{\mathscr{P}} =\begin{bsmallmatrix}
              \hat{P} & 0  \\
              0 & \hat{P}
             \end{bsmallmatrix} $ and $\hat{J} =
             \begin{bsmallmatrix}
               0 &I\\
               -I& 0
             \end{bsmallmatrix}
$ on $\hat{\mathcal{H}} = \mathcal{H}\oplus \mathcal{H}$. 
$\hat{S}-i \hat{J}\geq 0$ implies $ \hat{L}^{\tau}\hat{\mathscr{P}} \hat{L} - i  \hat{J} \geq 0$. By a conjugation with $\hat{L}^{-1}$ and using the fact that $L^{-1}$ is symplectic we get $\begin{bsmallmatrix}
              \hat{P} & 0  \\
              0 & \hat{P}
             \end{bsmallmatrix} - i  \begin{bsmallmatrix}
               0 &I\\
               -I& 0
             \end{bsmallmatrix} \geq 0$ on $\mathcal{H}\oplus \mathcal{H}$. Hence 
             we get $\begin{bsmallmatrix}
              \hat{P} & -iI  \\
              iI & \hat{P}
             \end{bsmallmatrix}\geq 0$ on $\CH \oplus \CH$. But this means $\hat{P} \geq I$ on $\CH$ and correspondingly $\mathscr{P}\geq I$ on $\mathcal{H}$.

\ref{item:8}. 
Since $S$ is positive and invertible, $\alpha(z,w):= \re \left\langle z, Sw \right\rangle$ defines a complete real  inner product on $\mathcal{H}$. Therefore, by (\ref{eq:1}),  a primary quasifree state $\phi$ as in (\ref{eq:52}) exists if and only if $\sigma(z,w)^{2} \leq \alpha(z,z)\alpha(w,w),$ for all $ f,g \in \mathcal{H}$. That is 
 \begin{equation}\label{eq:53}
 \im \left\langle z,w \right\rangle^2 \leq \re \left\langle z, Sz \right\rangle\re \left\langle w, Sw \right\rangle. 
\end{equation}
Now
\begin{align*}
{\im\left\langle z,w \right\rangle}^2 &= {\im\left\langle Lz,Lw \right\rangle}^2\\
                                   &\leq |\left\langle Lz,Lw \right\rangle|^{2}\\
                                  &\leq \left\langle Lz,Lz \right\rangle\left\langle Lw,Lw \right\rangle\\
                               &\leq \left\langle Lz,\mathscr{P}Lz \right\rangle\left\langle Lw,\mathscr{P}Lw \right\rangle \numberthis \label{eq:55.1}\\
                              & = \re \left\langle Lz,\mathscr{P}Lz \right\rangle\re\left\langle Lw,\mathscr{P}Lw \right\rangle\\
                              &=\left\langle z,L^{\tau}\mathscr{P}Lz \right\rangle\left\langle w,L^{\tau}\mathscr{P}Lw \right\rangle,
\end{align*}
where (\ref{eq:55.1}) follows from \ref{item:9}). Thus we proved (\ref{eq:53}). 
Further, $\phi = \phi_{A}$ because $\re \left\langle \cdot,S(\cdot) \right\rangle_{\mathcal{H}} = -\im \left\langle A(\cdot),\cdot \right\rangle $.
\end{proof}

\begin{lem}\label{second lemma}
  Let $S$ be a real linear symmetric and invertible  operator on $\mathcal{H}$. Suppose the function $f:\mathcal{H}\rightarrow \mathbb{R}$ defined by $f(z)= e^{-\frac{1}{2}\re\left\langle z,Sz \right\rangle}$ is the quantum characteristic function of a density operator $\rho$ i.e. $S \in \mathscr{K}_{\mathbf{G}}(\CH)$  then  on $(\mathcal{H}, \re \left\langle \cdot,\cdot \right\rangle)$,
\begin{enumerate}
\item \label{item:4} $S-I$ is Hilbert-Schmidt.
\item \label{item:5} $(\sqrt{S}J\sqrt{S})^{\tau}(\sqrt{S}J\sqrt{S})- I$ is in trace class.
\end{enumerate}
\end{lem}

\begin{proof}

\ref{item:4}.
We want to prove that $S-I$ is Hilbert-Schmidt on the real Hilbert space $\mathcal{H}$. 
Since $\hat{S}-i\hat{J} \geq 0$, by Lemma \ref{lem:positivity-condition-consequences} there exists a primary quasifree state $\phi$ on $CCR(\mathcal{H},\sigma)$ such that 
\begin{equation*}
\phi(W(z)) = e^{-\frac{1}{2}\re\left\langle z,Sz \right\rangle}.
\end{equation*}


Claim : $\phi_A$ and $\phi_{(-J)}$ are quasi equivalent, where $A=-JS$.

\begin{pf}[of Claim]
   Consider the state $\psi$ on $\B{\GH}$ given by $ X\mapsto \tr\rho X$. The  quasifree  state $\phi_A$ is the  restriction of $\psi$ to $\mathcal{A}: = CCR(\mathcal{H},\sigma) \hookrightarrow \B{\GH}$. Let $(H_\psi, \Pi_\psi, \Omega_\psi)$ be the GNS triple for $\B{\GH}$ with respect to $\psi$. Then $(H_\psi, \Pi_\psi|_{\mathcal{A}}, \Omega_\psi)$ is the GNS triple for $\mathcal{A}$ with respect to $\phi_A$.  To see this, only thing to be noticed is $\Omega_{\psi}$ is cyclic for $\Pi_{\psi}(\mathcal{A})$, which is clear since $\mathcal{A}$ is strongly dense in $\B{\GH}$.  We further note that the inclusion $\mathcal{A} \subseteq \B{\GH}$ is the GNS representation with respect to the vacuum state which is the quasi-free state given by $\phi_{-J}$. It can be seen that the association 
 $$W(x) \mapsto \Pi_{\psi}(W(x))$$
can be extended as an isomorphism between $\B{\GH} =\mathcal{A}'' $ and $\Pi_\psi(\B{\GH})$. Thus the claim is proved.
\end{pf}

  Since  $\phi_{(-J)}$ and $\phi_A$  are quasi equivalent, 
 by Theorem \ref{quasi-equivalence} 
  we get $A+J$ is Hilbert-Schmidt 
  on $\CH_{-J}$ which is same as $\CH$ with the real inner product $\re \left\langle \cdot,\cdot \right\rangle_{\mathcal{H}}$.   Since $A=-JS$, multiplying by $J$ on  $A+J$  we get $S-I$ is Hilbert-Schmidt on ($\CH, \re \left\langle \cdot,S(\cdot) \right\rangle_{\mathcal{H}}$). Now  the
  result is an easy consequence. 

\ref{item:5}. This follows due to the same reason as that of \ref{item:4}) because of Theorem \ref{quasi-equivalence} itself. We get $\sqrt{-A^2-I}$ is Hilbert-Schmidt on $(\CH, \re \left\langle \cdot,\cdot \right\rangle)$. This is same as  $-A^2-I$ is trace class on the same Hilbert space. Hence we have $-JSJS-I$ is trace class. By multiplying with $\sqrt{S}$ on the left and $(\sqrt{S})^{-1}$ on the right we see that $-\sqrt{S}JSJ\sqrt{S}- I$ is trace class. The result follows because $J^{\tau} = -J$


\end{proof}

   It may be noted at this point that the operator $\sqrt{S}J\sqrt{S}$ in \ref{item:5}) of the above theorem is the skew symmetric operator $B$ appearing in the proof of Williamson's normal form in \cite{BhJo18}. Proof of Williamson's normal form was obtained there by applying spectral theorem (as proved in  \cite{BhJo18}) to $B$, 
\[   \Gamma^{T} B \Gamma =  \begin{bmatrix}
                          0 & -P   \\
                          P & 0 
                         \end{bmatrix}  \]
where $\Gamma$ is an orthogonal transformation. $L$ was obtained by taking 
\[L = \begin{bmatrix}
                          P^{-1/2}  & 0  \\
                           0 &  P^{-1/2}
                         \end{bmatrix}\Gamma^{\tau}S^{1/2}.\] This choice of $L$ provides $S = L^{\tau}\mathscr{P}L$, where $\mathscr{P}_{0} =  \begin{bsmallmatrix}
                          P & 0   \\
                          0 & P 
                         \end{bsmallmatrix} $.
\begin{cor}
 Assume the hypothesis of the previous Lemma \ref{second lemma}. Then
\begin{enumerate}
\item \label{item:7}  If $S-I \geq 0$ then  $S-I$ is trace class on $(\mathcal{H}, \re \left\langle \cdot,\cdot \right\rangle)$
\item \label{item:6}If $S$ is complex linear then $S-I \geq 0$ and $S-I$ is trace class on $(\mathcal{H}, \re \left\langle \cdot,\cdot \right\rangle)$.
\end{enumerate}
\end{cor}
\begin{proof}
\ref{item:7}. We have  $-\sqrt{S}JSJ\sqrt{S}- I$ is trace class on ($\CH, \re \left\langle \cdot,\cdot \right\rangle $). Hence by multiplying with $(\sqrt{S})^{-1}$ on both sides $(-J)SJ- S^{-1}$ is trace class. Since $S-I\geq 0$,  $(-J)SJ - I \geq 0$ and $S^{-1}\leq I$ therefore we have $$0 \leq (-J)SJ- I \leq (-J)SJ- S^{-1}$$ and we conclude that $ (-J)SJ- I $ is trace class on ($\CH, \re \left\langle \cdot,\cdot \right\rangle $). Thus this part is proved by a conjugation with $J$.


\ref{item:6}. 
By \ref{item:8}) in Lemma \ref{lem:positivity-condition-consequences} 
we have 
 \begin{equation}
 \label{eq:22}
  -A^2-I \geq 0 \todo{\tiny check this because there is a bar involved in the definition of CCR. I find no problem but please pay more attention}
 \end{equation} with respect to the real inner product $\re \left\langle \cdot,S(\cdot) \right\rangle$.
We have $A^2= JSJS$ but since $S$ is complex linear it commutes with $J$, thus $A^2 =-S^2$ and we see that $S^{2}-I \geq 0$, consequently $S \geq I$ on  ($\CH, \re \left\langle \cdot,S(\cdot) \right\rangle_{\mathcal{H}}$). But this implies $S\geq I$ on ($\CH, \re \left\langle \cdot,\cdot \right\rangle$) since $S$ is positive. Since $S$ commutes with $J$, by \ref{item:5}) of Lemma \ref{second lemma} we see that $S^{2}-I$ is Hilbert-Schmidt on $(\mathcal{H}, \re \left\langle \cdot,\cdot \right\rangle)$. Now the claim follows because  $0 \leq S-I \leq S^2-I.$
\end{proof}

\begin{rmk}
  By \ref{shale-Gaussian}) of Example \ref{egs:Gaussian-states} we have seen that for a symplectic automorphism $L$, $L^{\tau}L$ is a covariance operator whenever $L^{\tau}L-I$ is Hilbert-Schmidt. Now by Lemma \ref{second lemma} we get that $L^{\tau}L$ satisfies the conditions \ref{item:4}, and \ref{item:5} there. This is true also for any such symplectic transformation. But  since $\sqrt{L^{\tau}L}$ is 
 symplectic  whenever $L$ is so, the condition \ref{item:5} is just void. Also it can be proved independently that for any symplectic transformation the positivity condition (\ref{eq:17}) on $L^{\tau}L $  is true. Therefore, $L^{\tau}L-I$ is Hilbert-Schmidt is the only non-trivial condition here.
\end{rmk}

\subsection{Sufficiency}
Now we proceed to prove that the conditions in Lemma \ref{first lemma} and \ref{second lemma}
are sufficient to ensure that $S$ is covariance operator of a quantum Gaussian state. This essentially involves an explicit construction of the quantum Gaussian state for positive complex linear operators and an application of Williamson normal form. We begin with some computations for this construction.
\begin{lem}\label{lem:analysis}
  If $s_j>0$ then $\sum\limits_{j=1}^{\infty}\left(\frac{e^{-s_j}}{1-e^{-sj}}\right) <\infty$ if and only if  $\sum\limits_{j=1}^{\infty}e^{-s_j}$ is convergent.
\end{lem}
\begin{proof}
  Assume  $\sum\limits_{j=1}^{\infty}\left(\frac{e^{-s_j}}{1-e^{-sj}}\right) <\infty$. Since $\frac{e^{-s_j}}{1-e^{-sj}}> 0$ and $\frac{1}{1-e^{-s_j}}>1$, we have $0<\sum\limits_{j=1}^{\infty} e^{-s_j}<\sum\limits_{j=1}^{\infty}\left(\frac{e^{-s_j}}{1-e^{-sj}}\right)<\infty$. Now assume that $\sum\limits_{j=1}^{\infty}e^{-s_j}<\infty$. Then $s_j\rightarrow\infty$ and hence $\frac{1}{1-e^{-s_j}}\rightarrow 1$. This means we have $0< \frac{1}{1-e^{-s_j}}<M,\forall j$, for some $M>1$. Therefore, $\sum\limits_{j=1}^{\infty}\left(\frac{e^{-s_j}}{1-e^{-sj}}\right) <\infty$.
\end{proof}
Let $\mathcal{H} = H +iH$ and $\{e_1,e_2,e_3\cdots\}$ be an orthonormal basis for $H$. Note that $\{e_j\}$ is also a basis for $\mathcal{H}$ as a complex Hilbert space.  Let  $D =\diag(d_j)$ be a  bounded diagonal operator on $\mathcal{H}$, with $d_j > 1$, $j = 1,2,3,\dots$ in the given basis. Since $d_j>1$ there exists $s_j>0$ such that $d_j=\coth(\frac{s_j}{2})$ for all $j$. If we consider $D$ as a real linear operator on $\mathcal{H}$, then $D_0 =
\begin{bsmallmatrix}
  D &0\\
  0 & D
\end{bsmallmatrix} $ on $H\oplus H$.

\begin{lem}\label{lem:gamma-st-trace-class}
Let $D =\diag(d_j)$ be a  bounded diagonal operator on $\mathcal{H}$, with $d_j > 1$, $j = 1,2,3,\dots$  with respect to a basis. Write  $d_j=\coth(\frac{s_j}{2})$ for all $j$.  Then $D-I$ is trace class if and only if $\sum\limits_{j=1}^{\infty}e^{-s_j}$ is convergent.
\end{lem}
\begin{proof} Observe, 
\begin{align*}
 D-I ~~\mbox{is in trace class} & \Leftrightarrow  \sum\limits_{j=1}^{\infty}(d_j-1) <\infty\\
                  & \Leftrightarrow \sum\limits_{j=1}^{\infty}(\coth(\frac{s_j}{2})-1) <\infty\\
                  & \Leftrightarrow \sum\limits_{j=1}^{\infty} \left( \frac{1+e^{-s_j}}{1-e^{-s_j}}-1 \right)\\
                  &\Leftrightarrow \sum\limits_{j=1}^{\infty}\left(\frac{e^{-s_j}}{1-e^{-sj}}\right) <\infty  \\
                  & \Leftrightarrow \sum\limits_{j=1}^{\infty}e^{-s_j} < \infty \numberthis \label{eq:1.1}
 \end{align*} 
where (\ref{eq:1.1}) follows from Lemma \ref{lem:analysis}.

\end{proof}
\begin{prop}\label{prop:Gaussian-state-exists}
   Let $D$ be as described in Lemma \ref{lem:gamma-st-trace-class} and $D-I$ is trace class.  Then there exists a state $\rho_D$ on $\GH$ such that $\hat{\rho}_D({x}) = e^{-\frac{1}{2}\left\langle x, Dx\right\rangle}$.
\end{prop}
\begin{proof}Consider the diagonal operator $T=\diag(e^{-s_j})$ with respect to the same basis in which $D$ is diagonal then the second quantization $\Gamma(T)$
is a trace class operator on the symmetric Fock space, $\Gamma_s(\mathcal{H})$ because of the following reasoning. $T$ is positive and by Lemma \ref{lem:gamma-st-trace-class} it is a trace class operator. Thus we have $s_j>0$ and $s_j\rightarrow\infty$. 
Therefore we get $\sup_j(e^{-s_j}) <1$. Now by Exercise 20.22 (iv) in \cite{Par12} , the second quantization (\ref{eq:73}) $\Gamma(T)$ is trace class with 
  \begin{equation}
\label{eq:3}
\tr\Gamma(T) = \Pi_{j=1}^{\infty}(1-e^{-s_j})^{-1}
\end{equation}
 Define $\rho_D =\Pi_{j=1}^{\infty}(1-e^{-s_j})\Gamma(T)$, then $\rho$ is a density operator on $\Gamma_s(\mathcal{H})$. We have $\CH = \oplus_j \mathbb{C}e_j$. Since $\Gamma(e^{-s_j}) =  e^{-s_ja_j^{\dagger}a_j}$ on $\Gamma_s(\mathbb{C}e_j)$, 
 $\rho_D = \Pi_{j=1}^{\infty}(1-e^{-s_j})\Gamma(\oplus_je^{-s_j})= \otimes_j\rho_j$, where $\rho_j = (1-e^{-s_j})e^{-s_{j}a_j^{\dagger}a_j}$.  Let $x = \oplus_{j} x_je_j$ then
\begin{align*}
\hat{\rho}_D(x) &= \tr \rho W(x)\\
              &= \tr \left(\Pi_{j=1}^{\infty}(1-e^{-s_j}) \Gamma(\oplus_je^{-s_{j}})W(\oplus_jx_j)\right)\\
              &= \tr \left(\otimes_j(1-e^{-s_j})\Gamma(e^{-s_{j}})W(x_j)\right)\\
              & =  \tr \left(\otimes_j(1-e^{-s_j})e^{-s_{j}a_j^{\dagger}a_{j}}W(x_j)\right)\\
              &= \Pi_{j=1}^{\infty} \tr \left((1-e^{-s_j})e^{-s_ja_j^{\dagger}a_{j}}W(x_j)  \right)\\
              &= \Pi_{j=1}^{\infty} e^{-\left\langle x_j,  \frac{1}{2}\coth(\frac{s_j}{2})x_j \right\rangle} \numberthis \label{eq:5}\\
              &= e^{-\frac{1}{2}\left\langle x, Dx \right\rangle}
 \end{align*} 
where (\ref{eq:5}) follows from Example \ref{eg:fundamental-mixed-gaussian} of Gaussian states.
\end{proof}
 

  

Recall from Example \ref{eg:coherent-state} that the vacuum state $\ketbra{e(0)}{e(0)}$ on $\Gamma_s(\CH)$ is a Gaussian state with covariance operator $I$.
\begin{prop}\label{thm:Gaussian-state-exists} If $\mathscr{P}$ is any complex linear operator on $\mathcal{H}$ such that $\mathscr{P}-I$ is positive and trace class then there exists a state $\rho$ on $\GH$ such that
the  quantum characteristic function  $\hat{\rho}$ associated with $\rho$ is given by
\begin{displaymath}
  \hat{\rho}({x}) = e^{-\frac{1}{2}\left\langle x, \mathscr{P}x\right\rangle}
\end{displaymath}
for every $x\in \mathcal{H}$.
\end{prop}
\begin{proof}Let $U$ be a unitary operator such that $\mathscr{P} = U^{*}DU$. Such a $U$ exists by applying spectral theorem to the compact positive operator $\mathscr{P}-I$. Since $\mathscr{P} \geq I$ assume without loss of generality that $\CH= \CH_1\oplus\CH_2$ is such that $D =
  \begin{bsmallmatrix}
    D_1 &0\\0 & I
  \end{bsmallmatrix}
$, where we separated all the diagonal entries of $D$ which are equal to one and not equal to one.  Then $D_1$ satisfies the assumptions in Proposition \ref{prop:Gaussian-state-exists} and $\rho_{D_1}$ exists  as a Gaussian state on $\Gamma_s(\CH_{1})$. Let $\rho_0$ denote the vacuum state $\ketbra{e(0)}{e(0)}$  on $\Gamma_s(\CH_{2})$, which is Gaussian by Example \ref{eg:coherent-state}. Then 
$\rho_{D_{1}} \otimes \rho_0 = \rho_g(0, D)$. Define $\rho = \Gamma_s(U^{*})\rho_{D_{1}} \otimes \rho_0\Gamma_s(U)$  and the result follows from Proposition \ref{prop:qcf-of-weyl-conjugation}.
\end{proof}

For a proof of the following Lemma see the discussion in the beginning of Section 2 in \cite{BhSr05}.
\begin{lem}
  \label{lem:inverse-sqrt-HS}
Let $C-I$ is Hilbert-Schmidt (trace class) then 
\begin{enumerate}
\item  If $C \geq 0$ then $\sqrt{C}-I$ is  Hilbert-Schmidt (trace class).
\item  If $C$ is invertible then $C^{-1}-I$ is  Hilbert-Schmidt (trace class).
\end{enumerate}
\end{lem}
\begin{lem}\label{HS-TC-conditions-consequences}
 Let $S$ be a real linear, positive and invertible operator on $\mathcal{H}$. Then $L$ and $\mathscr{P}$ as in Corollary \ref{cor:will-norm-form-complex-case} can be chosen such that 
\begin{enumerate}
\item\label{item:11} If $S-I$ is Hilbert-Schmidt then $L^{\tau}L -I$ is Hilbert Schmidt, \emph{i.e} $L \in  \mathscr{S}(\CH)$.
   \item \label{item:12}If $(\sqrt{S}J\sqrt{S})^{\tau}(\sqrt{S}J\sqrt{S})- I$ is trace class then $\mathscr{P} -I$ is a trace class operator on $\mathcal{H}$.
\end{enumerate}  
\end{lem}
\begin{proof}
 \ref{item:11}. It can be seen from the proof of Williamson's normal form in \cite{BhJo18} that $M_0$ in Theorem \ref{sec:will-norm-form} can be chosen as $M = \mathscr{P}^{-1/2}\Gamma^{\tau}S^{1/2}$, where $\Gamma_0 \colon K\oplus K \rightarrow H \oplus H$ is an orthogonal transformation such that the skew symmetric operator
\begin{equation}
\label{eq:26}
B_0:=S_0^{1/2}J_0S_0^{1/2}= \Gamma_0
\begin{bsmallmatrix}
  0 & -P\\
  P & 0
\end{bsmallmatrix} \Gamma_0^{\tau}.
\end{equation} and $\mathscr{P}_{0} =
\begin{bsmallmatrix}
  P & 0 \\
  0 & P
\end{bsmallmatrix}
$ (recall the notation we employed at the beginning of \ref{sec:sympl-autom} to distinguish between the two avatar's of a real linear operator defined on a complexified Hilbert space). Since $L = UM$ in the proof of Corollary \ref{cor:will-norm-form-complex-case}, 
\begin{equation}
\label{eq:27}
L^{\tau}L = S^{1/2}\Gamma\mathscr{P}^{-1}\Gamma^{\tau}S^{1/2}.
\end{equation}
But 
\begin{equation}
\label{eq:28}\begin{bsmallmatrix}
  0 & P\\
  -P & 0
\end{bsmallmatrix}
\begin{bsmallmatrix}
  0 & -P\\
  P & 0
\end{bsmallmatrix} = \begin{bsmallmatrix}
  P^2 & 0\\
  0 & P^2
\end{bsmallmatrix} ,
\end{equation}
therefore if  we write $\mathsf{P}_0 =\begin{bsmallmatrix}
  0 & -P\\
  P & 0
\end{bsmallmatrix} $, and recall the relationship between the operator $\mathsf{P}_0$ on the real Hilbert space and $\mathsf{P}$ on the complexified Hilbert space (described at the beginning of \ref{sec:sympl-autom}), we see that 
\begin{equation}
\label{eq:29}
\mathscr{P}^{-1} = (\sqrt{\mathsf{P}^{\tau}\mathsf{P}})^{-1}
\end{equation}
Since $\Gamma$ is orthogonal, by (\ref{eq:26}) and (\ref{eq:29}) we get $(\sqrt{B^{\tau}B})^{-1} = \Gamma \mathscr{P}^{-1} \Gamma^{\tau} $. Now by (\ref{eq:27}) we get 
\begin{equation}
\label{eq:30}
L^{\tau}L = S^{1/2}(\sqrt{B^{\tau}B})^{-1}S^{1/2}.
\end{equation}
We have $S-I$ is Hilbert-Schmidt. Therefore, so is $J^{\tau}SJ- I$. Hence $S^{1/2}J^{\tau}SJS^{1/2} - S$ is Hilbert-Schmidt. By adding and subtracting $I$ and using the fact the $S- I$ is Hilbert-Schmidt we get $S^{1/2}J^{\tau}SJS^{1/2} - I$ is also so. In other words, we just got $B^{\tau}B- I $ is Hilbert-Schmidt. Now by Lemma \ref{lem:inverse-sqrt-HS} we get $(\sqrt{B^{\tau}B})^{-1}-I$ is Hilbert-Schmidt. This along with (\ref{eq:30}) finally allows us to conclude that $L^{\tau}L - I$ is Hilbert-Schmidt.

\ref{item:12}. By keeping the notations above and using Lemma \ref{lem:inverse-sqrt-HS}, we have $(\sqrt{B^{\tau}B})^{-1}-I$ is trace class and thus $S^{1/2}(\sqrt{B^{\tau}B})^{-1}S^{1/2}-S = L^{\tau}L -S$ is trace class. Since $S = L^{\tau}\mathscr{P}L$ we get $L^{\tau}(\mathscr{P}-I)L$ is trace class. Since $L$ is invertible we see that $\mathscr{P}-I$ is trace class.
\end{proof}

Here is the theorem which characterizes all covariance operators of quantum Gaussian states. 
It is essentially the statement that there exists a quantum Gaussian state $\rho$ with covariance matrix $S$ on the Boson Fock space of $\CH$  if and only if $\hat{\rho}_{|_{CCR(\CH, \sigma)}}$ is a  primary quasifree state $\phi_A$  quasiequivalent to the vacuum state $\phi_{-J}$ on $CCR(\CH, \sigma)$, where $A=-JS$.

\begin{thm}\label{thm:main-thm}
  Let $S$ be a real  linear, bounded, symmetric and invertible  operator on $\mathcal{H}$. Then $S$ is the covariance operator of a quantum Gaussian state if and only if the following holds
\begin{enumerate}
\item\label{item:16}
$\hat{S}-i \hat{J}\geq 0$  on $\hat{\CH}$.
\item\label{item:17}$S-I$ is Hilbert-Schmidt on $(\mathcal{H}, \re \left\langle \cdot,\cdot \right\rangle)$.
\item\label{item:18} $(\sqrt{S}J\sqrt{S})^{\tau}(\sqrt{S}J\sqrt{S})- I$ is trace class on $(\CH, \re \left\langle \cdot,\cdot \right\rangle)$. 
\end{enumerate}
\end{thm}

\begin{proof}
In the last Section in Lemma \ref{first lemma} and \ref{second lemma} we have seen the necessity of conditions 1-3 for a quantum Gaussian state. Now we prove sufficiency. So  assume 1-3. As  $\hat{S}-i\hat{J} \geq 0$, $S\geq 0$ and by assumption $S$ is invertible and hence it has Williamson's normal form. Thus there exists a symplectic automorphism $ L$ on $\mathcal{H}$ such that  $S = L^{\tau}\mathscr{P}L$ (Corollary \ref{cor:will-norm-form-complex-case}).
By Lemma \ref{HS-TC-conditions-consequences}$, \mathscr{P} - I$ is trace class and $L^{\tau}L - I$ is Hilbert-Schmidt. Now by Theorem \ref{thm:generalized-shale-unitary} there exists a unitary operator $\Gamma_s(L)$ on $\Gamma_s(\mathcal{H})$ such that 
\begin{equation}
\label{eq:31}
\Gamma_s(L)W(u)\Gamma_s(L)^{*} = W(Lu)
\end{equation}

Since $\mathscr{P}- I$ is trace class and positive, by Proposition \ref{thm:Gaussian-state-exists}  there exists a density operator $\rho_{\mathscr{P}}$ such that $\hat{\rho}_{\mathscr{P}}(y) = e^{-\frac{1}{2} \left\langle y, \mathscr{P}y \right\rangle}$ for every $y \in \mathcal{H}$. Define 
\begin{equation}
\label{eq:32}
\rho = \Gamma_s(L)^{*}\rho_{\mathscr{P}}\Gamma_s(L)
\end{equation}
\begin{claim}
$\hat{\rho}(z) = e^{-\frac{1}{2}\re\left\langle z,Sz \right\rangle}$ for every $z\in \mathcal{H}$.
\end{claim}
\begin{pf}[of Claim] By Proposition \ref{prop:qcf-of-weyl-conjugation} we have
 \begin{align*}
\hat{\rho}(z)
              & = \hat{\rho_{\mathscr{P}}}(Lz)\\ 
              & = e^{-\frac{1}{2} \left\langle Lz, \mathscr{P}Lz \right\rangle}\\
              & = e^{-\frac{1}{2} \re \left\langle Lz, \mathscr{P}Lz \right\rangle}\\
              & = e^{-\frac{1}{2} \re \left\langle z, L^{\tau}\mathscr{P}Lz \right\rangle}\\
              & = e^{-\frac{1}{2} \re \left\langle z, Sz \right\rangle .}
  \end{align*} \end{pf}
\end{proof}

It easily seen that if $S$ is complex linear then the third condition of previous theorem is redundant. So we have the following.
\begin{cor} \label{cor:spl-case-C-linear}
    Let $S$ be a complex linear, self-adjoint and invertible  operator on $\mathcal{H}$. Then $S$ is the covariance operator of a quantum Gaussian state on $\GH$ if and only if $\hat{S}-i\hat{J} \geq 0$ and $S-I$ is trace class.
\end{cor} 
\begin{cor}\label{cor:S-gereater-I}
  Let $S\geq I$ be real linear then  $S$ is the covariance operator of a quantum Gaussian state on $\GH$ if and only if  $S-I$ is trace class
\end{cor}
\begin{rmk}
Notice that even in the finite mode case the condition $\hat{S}-i\hat{J} \geq 0$ does not necessarily imply that $S \geq I$. For example, one can consider the $2\times2$ matrix $\bmqty{2 &0\\0&\frac{1}{2}} $, which is a valid covariance matrix  but not greater than $I$.
\end{rmk}
\section{Convexity Properties of  Covariance Operators}

The previous section described and characterized infinite mode Gaussian states. The Hilbert Schmidt conditions appearing there show that unlike finite dimensions covariance operators in infinite dimensions do not form a cone. However, the next proposition shows that they do form a convex set. This allows us to extend some beautiful symmetry properties of Gaussian states proved by Parthasarathy \cite{Par10, Par13}  in the finite mode case to this setting. 
 In the following  $\mathscr{K}_{\mathbf{G}}(\CH)$ denotes the collection of  covariance operators for Gaussian states on $\GH$ and  $\mathscr{S}(\CH)$ are Shale operators on $\mathcal {H}$ for a fixed infinite dimensional separable real Hilbert space $\mathcal{H}$. The following Proposition is an easy  consequence of earlier discussions.

\begin{prop} 
  Consider two mean zero Gaussian states \[\rho_i = \rho_g(0,S_i), i = 1,2  \] on $\GH$. For $\theta \in \mathbb{R}$, let  $U_{\theta}$ be the unitary operator 
$\begin{bsmallmatrix}
   \cos \theta & -\sin \theta\\ 
    \sin \theta & \cos \theta
  \end{bsmallmatrix}$ on $\CH \oplus \CH$. Then \[\tr_2 \big(\Gamma_s(U_{\theta})(\rho_1\otimes\rho_2) \Gamma_s(U_{\theta})^{*}\big) = \rho_g(0, (\cos^2\theta ) S_1+(\sin^2\theta ) S_2) \] where $\tr_2$ denotes the relative trace over the second factor of $\Gamma_s(\CH)\otimes \Gamma_s(\CH)$.
   Consequently $\mathscr{K}_{\mathbf{G}}(\CH)$ is a convex set.
\end{prop}


\begin{lem}
  Let $P \geq I$, then there exists invertible positive operators $P_{1}$ and $P_2$ such that 
\begin{equation}
\label{eq:102}
P= \frac{1}{2}(P_1+P_2)=\frac{1}{2}(P_1^{-1}+P_2^{-1})
\end{equation}
\end{lem}
\begin{proof}
  Take $P_1= P+\sqrt{P^2-I}$ and $P_2 = P-\sqrt{P^2-I}$. Then $P_{1}P_2=P_2P_1=I$ and (\ref{eq:102}) is satisfied.
\end{proof}
\begin{lem}\label{sec:conv-prop-covar}
  Let $\CH=H+iH$ and $\mathscr{P} \in \B{\CH}$ be such that  $\mathscr{P}-I$ is positive and trace class, further let $\mathscr{P}_0 = \begin{bsmallmatrix}
  P &0\\
  0 & P
\end{bsmallmatrix}$ on $H\oplus H$(~\ref{sec:sympl-autom}). Then $\mathscr{P} = \frac{1}{2}(\mathscr{P}_1+\mathscr{P}_2)$, for some $ \mathscr{P}_j \geq 0$, and  $\mathscr{P}_j^{\frac{1}{2}} \in \mathscr{S}(\CH), j= 1,2$. 
\end{lem}
\begin{proof}
  Take $P_1 = P+\sqrt{P^2-I} $ and $P_2= P-\sqrt{P^2-I}$, then by (\ref{eq:102}) 
\begin{align*}
\mathscr{P}_0 =\frac{1}{2}\left\{ \begin{bmatrix}
  P_1 &0\\
  0 & P_1^{-1}
\end{bmatrix} +
      \begin{bmatrix}
         P_2 &0\\
         0 & P_2^{-1}
      \end{bmatrix}\right\}
\end{align*}
Define $\mathscr{P}_j$ such that $\mathscr{P}_j(x+iy) = P_jx+P_j^{-1}y, \forall x, y \in H, j=1,2$. Then $\mathscr{P}_j$ is symplectic and positive. To prove $\mathscr{P}_j^{\frac{1}{2}} \in \mathscr{S}(\CH)$ we should prove $\mathscr{P}_j -  I $ is Hilbert-Schmidt. To this end, it is enough to show that $P_j - I$ is Hilbert-Schmidt, $j =1,2$. Since $P - I$ is trace class (and hence Hilbert-Schmidt) it is enough to show $\sqrt[]{P^2 - I}$ is Hilbert-Schmidt or equivalently $P^2-I$ is trace class. This is true because  $P^2-I = (P-I)^2+ 2(P-I)$. 
\end{proof}
\begin{thm}\label{sec:conv-prop-covar-1}
  $S \in \mathscr{K}_{\mathbf{G}}(\CH)$ if and only if 
\begin{equation}
\label{eq:104}
S= \frac{1}{2}(N^{\tau}N+M^{\tau}M)
\end{equation} for some $N, M \in \mathscr{S}(\CH)$. Further, $S$ is an extreme point of $\mathscr{K}_{\mathbf{G}}(\CH)$ if and only if $S=N^{\tau}N$ for some $N \in \mathscr{S}(\CH)$.
\end{thm}
\begin{proof}
 Note that if $N\in \mathscr{S}(\CH) $ then $N^{\tau}N$ is a covariance operator (see \ref{shale-Gaussian} of Example  \ref{egs:Gaussian-states}).
Now let $S \in \mathscr{K}_{\mathbf{G}}(\CH)$, let $S = L^{\tau}\mathscr{P}L$ be the Williamson's normal form as in Corollary \ref{cor:will-norm-form-complex-case}. Then by Lemma \ref{HS-TC-conditions-consequences} and Lemma \ref{lem:positivity-condition-consequences} we have $L\in \mathscr{S}(\CH)$ and $\mathscr{P}-I$ is trace class and positive. By Corollary \ref{cor:will-norm-form-complex-case} we have $\mathscr{P}_0=
\begin{bsmallmatrix}
  P &0\\
  0 & P
\end{bsmallmatrix}
$. By Lemma \ref{sec:conv-prop-covar}, $\mathscr{P} = \frac{1}{2}(\mathscr{P}_1+\mathscr{P}_2)$ with $\mathscr{P}_j\geq 0, j =1,2$. Therefore we have \[S = \frac{1}{2}L^{\tau}(\mathscr{P}_1+\mathscr{P}_2)L .\]
By taking $N = \mathscr{P}_1^{1/2}L$ and $M=\mathscr{P}_2^{1/2}L$ we get (\ref{eq:104}). An easy computation shows $N, M \in\mathscr{S}(\CH)$. \\
Proof 
   of second part of the Theorem goes in an exactly similar lines to the proof of the similar statement in the finite mode case, Theorem 3 in \cite{Par13}. We give it here for completeness.
The first part also shows that for an element $S$ of $\mathscr{K}_{\mathbf{G}}(\CH)$ to be extremal it is necessary that $S = L^\tau L$ for some $L \in \mathscr{S}(\CH)$. To prove sufficiency, suppose there exist $L \in \mathscr{S}(\CH)$ and $S_1,S_2 \in \mathscr{K}_{\mathbf{G}}(\CH)$ such that 
\[L^\tau L = \frac{1}{2}(S_1+S_2).\]
By the first part of the theorem there exist $L_j\in \mathscr{S}(\CH)$ such that \begin{equation}\label{eq:104.a}
L^\tau L = \frac{1}{4}\sum\limits_{j=1}^4L_j^\tau L_j    
\end{equation}
where $S_1 = \frac{1}{2}(L_1^\tau L_1+L_2^\tau L_2)$, $S_2 = \frac{1}{2}(L_3^\tau L_3+L_3^\tau L_3)$. Left multiplication by $(L^\tau )^{-1}$ and right multiplication by $L^{-1}$ on both sides of (\ref{eq:104.a}) gives 
\begin{equation}\label{eq:104.b}
    I = \frac{1}{4}\sum\limits_{j=1}^4M_j
\end{equation}
where $M_j = (L^\tau )^{-1}L_j^\tau L_jL^{-1}.$ Each $M_j \in \mathscr{S}(\CH)$ and is positive. Multiplying by $J$ on both sides of (\ref{eq:104.b}) we get
$$    J = \frac{1}{4}\sum\limits_{j=1}^4M_jJ      = \frac{1}{4}\sum\limits_{j=1}^4M_jJM_jM_j^{-1}
      = \frac{1}{4}J\sum\limits_{j=1}^4M_j^{-1}.$$
Thus \[I = \frac{1}{4}\sum\limits_{j=1}^4M_j = \frac{1}{4}\sum\limits_{j=1}^4M_j^{-1}=\frac{1}{4}\sum\limits_{j=1}^4 \frac{1}{2}(M_j+M_j^{-1}),\]
which implies \[\sum\limits_{j=1}^4(M_j^{1/2}-M_j^{-1/2})^2 = 0,\]
or $M_j = I, \phantom{...}\forall 1\leq j \leq 4.$
Thus $L_j^\tau L_j = L^\tau L, \phantom{...}  \forall j$ and $S_1 = S_2$.
\end{proof}

\begin{cor}
  Let $S_1$, $S_2$ be extreme points of $\mathscr{K}_{\mathbf{G}}(\CH)$ such that $S_1\geq S_2$. Then $S_1=S_2$.
\end{cor}
\begin{proof}
\todo[inline]{Please pay attention to the part where $\mathcal{K}_1$ and $\mathcal{K}_{2}$ are identified to make sense of $L_2L_1^{-1}$. Hope I didn't make a mistake!}
  By Theorem \ref{sec:conv-prop-covar-1}, let $S_1=L_1^{\tau}L_1$ and  $S_{2} = L_2^{\tau}L_2$ for some $L_{1}, L_2 \in \mathscr{S}(\CH) $.
  $L_1^{\tau}L_1 \geq L_2^{\tau}L_2$ implies that the symplectic transformation $M:= L_2L_1^{-1}$ has the property $M^{\tau}M\leq I$. But since $M^{\tau}M$ is a positive symplectic automorphism $M^{\tau}M = VTV^*$ for some unitary $V$, where $\CH = H+iH$ and $T(x+iy) = Ax+iA^{-1}y$ for some positive invertible operator $A$ on $H$. This can be seen by applying Proposition \ref{prop:factoring-symplectic-transf} to $M$. But such a $T\leq I$ if and only if $A=I$. This proves $M^{\tau}M = I$. But this implies $L_{2}^{\tau}L_2=L_1^{\tau}L_1$ from the definition of $M$.
\end{proof}
\section{Structure of Quantum Gaussian States}

If $S$ is a Gaussian covariance matrix it satisfies the properties listed in Theorem \ref{thm:main-thm}, then by combining Lemma \ref{lem:positivity-condition-consequences} and  Lemma \ref{HS-TC-conditions-consequences} we get a Williamson's normal form (Corollary \ref{cor:will-norm-form-complex-case}), $S = L^{\tau}\mathscr{P}L$ such that $\mathscr{P}-I$ is positive and trace class. By spectral theorem, there exists a unitary $U$ such that  $\mathscr{P} = U^{*}DU$, where $D$ is diagonal and positive. Since a unitary  is also symplectic, whenever  $S$ is a covariance operator we may assume without loss of generality that the $\mathscr{P}$ occurring in the Williamson's normal form is of the form $\mathscr{P} =
  \begin{bsmallmatrix}
    D & 0\\
    0 & I
  \end{bsmallmatrix}
$ on a decomposition $\mathcal{H} = \mathcal{H}_1\oplus \mathcal{H}_2$, with $D = \diag(d_1,d_2,\dots)$, $d_{1}\geq d_{2}\geq \dots >1$. But now we have fixed a basis of $\mathcal{H}$, therefore through the identification of $\Gamma_s(\mathbb{C})$ with $L^2(\mathbb{R})$ (Lebesgue measure), ($e(z) \in \Gamma_s(\mathbb{C})$ is identified with the $L^2$-function $x \mapsto (2\pi)^{-1/4}\exp{-4^{-1}x^2+zx-2^{-1}z^{2}}$ (refer Example \ref{eg:fock-L2-iso}
), we can assume without loss of generality that $\Gamma_s(\CH)= \otimes_jL^2(\mathbb{\mathbb{R}})$, with respect to the stabilising vector $e(0)$.  

\begin{thm}\label{thm:struct-quant-gauss}
  Let $\rho_g(w,S)$ be a Gaussian state in $\GH$. Let $S =  L^{\tau}\mathscr{P}L$ be a Williamson's normal form of $S$,  with $L^{\tau}L - I$ is Hilbert-Schmidt and $\mathscr{P} =
  \begin{bsmallmatrix}
    D & 0\\
    0 & I
  \end{bsmallmatrix}
$,  on a decomposition $\mathcal{H} = \mathcal{H}_1\oplus \mathcal{H}_2$, with $D = \diag(d_1,d_2,\dots)$, $d_{1}\geq d_{2}\geq \dots >1$,   
   $d_j = \coth(\frac{s_j}{2})$, $\forall j$ . Then 
\begin{equation}
\label{eq:67}
\rho_g(w,S) = W(\frac{-i}{2}w)^{*}\Gamma_s(L)^{*}[\otimes_j(1-e^{-s_j})e^{-s_ja_{j}^{\dagger}a_j}\otimes \rho_0 ]\Gamma_s(L)W(\frac{-i}{2}w).
\end{equation}
where $\rho_0 = \ketbra{e(0)}{e(0)}$ is the the vacuum state on $\Gamma_s(\mathcal{H}_2)$. 

\end{thm}
\begin{proof}
  By Proposition \ref{prop:weyl-conjugation-Gaussian-state}, 
 $\rho_g(w,S) = W(\frac{-i}{2}w)^{-1}\rho_g(0,S)W(\frac{-i}{2}w) $.  Since $S =  L^{\tau}\mathscr{P}L$, 
 $\rho_g(0,S) = \Gamma_s(L)^{*}\rho_g(0,\mathscr{P})\Gamma_s(L)$. Since $\mathscr{P} = D \oplus I$, 
 $\rho_g(0,\mathscr{P}) = \rho_g(0,D) \otimes \rho_g(0,I)$.  But $\rho_g(0,D) = \otimes_j(1-e^{-s_j})e^{-s_ja_{j}^{\dagger}a_j}$.  Since both on left and right hand sides have same quantum characteristic function by proof of Proposition \ref{prop:Gaussian-state-exists} and it is obvious that $ \rho_g(0,I)= \rho_0 $. 
\end{proof}

\begin{cor}\label{cor:struct-pure-quant-gauss}
  If $\{e_j\}$ is an orthonormal  basis of $\CH$, consider $\GH= \Gamma (\oplus _j\mathbb{C}e_j )= \otimes_j L^2(\mathbb{R}).$ Then the \todo{should define this?}wave  function   of a general pure quantum Gaussian state is of the form 
\begin{equation}
\label{eq:68}
\ket{\psi}= W(\alpha) \Gamma_s(U)(\otimes_j\ket{e_{\lambda_j}} ) 
\end{equation}
where $e_{\lambda}\in L^2(\mathbb{R})$ is defined for $\lambda >0$ by 
\begin{equation*}
  e_{\lambda}(x) = (2\pi)^{-1/4}\lambda^{-1/2}\exp{-4^{-1}\lambda^{-2}x^2}, \phantom{...} x \in \mathbb{R};
\end{equation*}
$\alpha \in \CH$, $U$ is a unitary operator on $\CH$, $\Gamma_s(U)$ is the second quantization unitary operator associated with $U$ and $\lambda_j, j\in \mathbb{N}$ are positive scalars.
\end{cor}
\begin{proof}
  The proof is essentially similar to the proof of Corollary 2 in \cite{Par13} because of Theorem \ref{thm:struct-quant-gauss} and Proposition \ref{prop:factoring-symplectic-transf}.
\end{proof}


\begin{thm}
  [\textbf{Purification}] Let $\rho$ be a mixed Gaussian state in $\GH$. Then there exists a pure Gaussian state $\ket{\psi}$ in $\GH \otimes \GH$ such that \[\rho= \tr_2 U\ketbra{\psi}{\psi}U^{*} \] where $U$ is a unitary and $\tr_2$ is the relative trace over the second factor.
\end{thm}
\begin{proof}
  Proof is same as that of Theorem 5 in \cite{Par13}.
\end{proof}

 
\section{Symmetry group of Gaussian states}
  Let $\mathcal{H}$ be a complex separable infinite dimensional Hilbert space and let $\mathbf{G}(\CH)$ denote the set of all Gaussian states on $\GH$. In this Section, we characterize all automorphisms of $\B{ \GH }$ preserving the set of Gaussian states. 
\begin{defn}
 A unitary operator $U$ on $\GH$ is called a \emph{Gaussian symmetry} if  $U\rho U^{*} \in \mathbf{G}(\CH)$ for every $\rho\in \mathbf{G}(\CH)$. 
\end{defn}
We use $\mathbb{Z}_+$ to denote the set $\{0,1,2,3,\dots \}$ and take $\mathbb{Z}_{+,0}^{\infty}:=\{(k_1,k_2,\dots, k_n,\allowbreak 0,0,\dots)^{\tau}| k_j \in \mathbb{Z}_+, j, n\in \mathbb{N}\}$. Let $\{e_j\}_{j\in \mathbb{N}}$ denote the standard orthonormal basis for $\ell^2(\mathbb{N})$, the column vector with $1$ at the $j^{th}$ position and zero elsewhere.  An infinite order matrix $A$ is said to be a permutation matrix if $A$ corresponds to a unitary operator which maps $\{e_j\}$ to itself. 
\begin{lem}\label{lem:basic-technical}
  Let $\{s_j\}_{j\in \mathbb{N}}$ and $\{t_j\}_{j\in \mathbb{N}}$ be two sequences consisting of positive real numbers such that
  their positive integer linear combinations form same sets, that is, 
 \begin{equation}\label{eq:69}
\left \{ \sum\limits_{j=1}^ns_jk_j \vert k_j\in \mathbb{Z}_+ \forall j, n\in \mathbb{N}\right\} = \left \{ \sum\limits_{j=1}^nt_jk_j \vert k_j\in \mathbb{Z}_+\forall j, n\in \mathbb{N}\right\}.
\end{equation}
If $\{s_j\}_{j\in \mathbb{N}}$ and $\{t_j\}_{j\in \mathbb{N}}$ are both linearly independent over the field $\mathbb{Q}$, then 
there exists a bijection $\sigma: \mathbb{Z}_+\to \mathbb{Z}_+$ such that
$s_j = t_{\sigma (j)}$ for all $j\in \mathbb{Z}_+.$
\end{lem}
\begin{proof}
Consider $s = (s_1,s_2,s_3,\dots)^{\tau}$ and $t = (t_1,t_2,t_3,\dots)^{\tau}$ 
as vectors in $\mathbb{R}^{\infty}.$ The condition  \eqref{eq:69} in particular means that each $t_i$ is a positive integer linear combination of $\{s_j\}_{j\in \mathbb{N}}.$ Hence there exists an infinite matrix $A$ with entries from
$\mathbb{Z}_+$ such that
$$t_i= \sum _{j=1}^{\infty}a_{ij}s_j ~~\forall i.$$
Similarly there exists another infinite matrix $B$ with entries from
$\mathbb{Z}_+$ such that
$$s_i= \sum _{j=1}^{\infty}b_{ij}t_j ~~\forall i.$$
Note that by construction every row of $A,B$ have only finitely many non-zero entries. From this, it is easily verified that
formal matrix multiplications $AB$ and $BA$ make sense and give us matrices with entries from $\mathbb{Z}_+$ of same kind, namely every row has only finitely many non-zero entries. We also get
$$BAs=s,~~ ABt=t.$$
Then, by the assumption of linear independence on rationals of $\{s_j\}_{j\in \mathbb{N}}$ and $\{t_j\}_{j\in \mathbb{N}},$ we get $AB=BA=I$. In other words $A$ and $B$ are permutation matrices and inverses of each other. Clearly this completes the proof.
\end{proof}

Now we develop some notation which would help us in our computations. Recall the action of annihilation and creation operators $a, a^{\dagger}$  on  $\Gamma_s(\mathbb{C})$, from  Exercise 20.18(b) in \cite{Par12}: $\Gamma_s(\mathbb{C})$  has a complete orthonormal basis $\{\ket{k} : k\in \mathbb{Z}_+\}$, with $\ket{0}$ as the vacuum vector, $a(\ket{k})= \sqrt{k}\ket{k-1}, a(\ket{0})= 0,
a^{\dagger}(\ket{k})= \sqrt{k+1}\ket{k+1}$ Further each $\ket{k}$ is an eigenvector for the number operator $a^{\dagger}a$ with eigenvalue $k$, that is, $a^{\dagger}a(\ket{k})= k\ket{k}.$

Now let $\CH$ be an infinite dimensional Hilbert space with orthonormal basis, $\{e_j\} _{j\in \mathbb{N}}$. Then we may identify   $\Gamma_s(\mathcal{H})$ with  $\otimes_{j=1}^{\infty}\Gamma_s(\mathbb{C}e_{j})$, where the sequence of vacuum vectors is chosen as the stabilizing sequence. This identification can be done with respect to any orthonormal basis of $\CH .$ 
Define 
\begin{equation}
\label{eq:85}
\ket{\mathbf{k}} = \ket{k_1}\otimes \ket{k_2} \otimes \cdots \otimes \ket{k_N}\otimes \ket{0}\otimes \ket{0} \otimes \cdots=:\ket{k_1} \ket{k_2} \cdots \ket{k_N}
\end{equation} for $\mathbf{k}= (k_1,k_2,k_3, \dots)^{\tau} \in \mathbb{Z}_{+, 0}^{\infty}$. 
%
It can be seen that $\{\ket{\mathbf{k}}| \mathbf{k}\in \mathbb{Z}_{+,0}^{\infty}\}$ forms an orthonormal basis for $\Gamma_s(\mathcal{H})$.  For $u\in \mathcal{H}$, taking $u_j=\langle e_j, u\rangle, $ 
$$
\braket{\mathbf{k}}{e(u)}  
= \Pi_{j=1}^n \frac{u_j^{k_j}}{\sqrt{k_j!}}:= \frac{u^{\mathbf{k}}}{\sqrt{\mathbf{k}!}}, 
$$
where the last equality defines the multi-index notation and also we take $0^0 = 1$. Therefore we write, 
\begin{equation}
\label{eq:101}
e(u)= \sum\limits_{\mathbf{k}\in \mathbb{Z}_{+,0}^{\infty}}^{}\frac{u^{\mathbf{k}}}{\sqrt{\mathbf{k}!}}\ket{\mathbf{k}}.
\end{equation}

We also have,

\begin{equation}
\label{eq:86}
(I\otimes I \otimes \cdots \otimes I \otimes a_j^{\dagger}a_j\otimes I \otimes I \otimes \cdots) (\ket{\mathbf{k}}) =
\begin{cases*}
  k_j\ket{\mathbf{k}}, & if $j\leq N$\\
  0, & otherwise
\end{cases*},
 \end{equation} where $a_j^{\dagger}a_j$ is the number operator on $\Gamma_s(\mathbb{C}e_{j}), j \in \mathbb{N}$. \\

Consider $\mathcal{H} = \oplus_{j=1}^{\infty} \mathbb{C}e_{j}$ as above. For a sequence of positive numbers $\{s_j\}_{j\in \mathbb{N}}$ such that $d_j = \coth(\frac{s_j}{2}) > 1$ and $\sum\limits_j (d_j-1)$ is finite, we know from Theorem \ref{thm:struct-quant-gauss} 
that, there exists a Gaussian state $\rho_s= \Pi_{j=1}^{\infty}(1-e^{-s_j})\otimes_{j=1}^{\infty}e^{-s_ja_{j}^{\dagger}a_j} \in \B{\GH}$. Then we have the  
\begin{lem}\label{lem:symm-group}
The spectrum of the  Gaussian state $\rho_s$
is,
\begin{equation}
\label{eq:87}
\sigma (\rho_s) =  \left\{pe^{-\sum\limits_{j=1}^N s_jk_j}\big| k_j \in \mathbb{Z}_{+}, N \in \mathbb{N}\right\} \cup \{0\},
\end{equation}
where  $p:=\Pi_{j=1}^{\infty}(1-e^{-s_j})$. The point spectrum $\sigma_{\mathtt{p}}(\rho_s) = \sigma (\rho_s)\setminus \{0\}$. Further, if  $\{s_j\}_{j\in \mathbb{N}}$ is a  sequence of (distinct) irrational  numbers which are linearly independent over the field $\mathbb{Q}$ then each number $pe^{-\sum\limits_{j=1}^N s_jk_j}$ is an eigenvalue with multiplicity one.
\end{lem}

\begin{proof} Since $\rho_s$ is is a compact operator $\sigma (\rho_s) = \sigma_{\mathtt{p}}(\rho_s) \cup \{0\}$. We will show that $\sigma_{\mathtt{p}} (\rho_s) = \{pe^{-\sum\limits_{j=1}^N s_jk_j}\big| k_j \in \mathbb{Z}_{+}, N \in \mathbb{N}\}$.
  We have $\otimes_{j=1}^{\infty}e^{-s_ja_{j}^{\dagger}a_j} :=\slim_{N\rightarrow\infty} \allowbreak \otimes_{j=1}^Ne^{-s_ja_{j}^{\dagger}a_j}\otimes I\otimes I \otimes \cdots$. 
Therefore, $\otimes_{j=1}^{\infty}e^{-s_ja_{j}^{\dagger}a_j}(e(u) \otimes e(0)\otimes e(0)\otimes\cdots) = \left(\otimes_{j=1}^Ne^{-s_ja_{j}^{\dagger}a_j}e(u)\right)\otimes e(0)\otimes e(0)\otimes\cdots $, $\forall u \in \mathbb{C}^{N}$.
Thus  $\Gamma_s(\mathbb{C}^{N})$ is a reducing subspace for $\rho_s$ and
${\rho_s}_{|_{\Gamma_s(\mathbb{C}^{N})}} = \otimes_{j=1}^Ne^{-s_ja_{j}^{\dagger}a_j}$,  $\forall N$. Therefore,  
\begin{equation}
\label{eq:90}
\rho_s(\ket{\mathbf{k}}) = (p\otimes_{j=1}^{\infty}e^{-s_ja_{j}^{\dagger}a_j})\ket{\mathbf{k}} = pe^{-\sum\limits_{j=1}^{\infty} s_jk_j}\ket{\mathbf{k}}, \forall \mathbf{k} \in \mathbb{Z}_{+,0}^\infty.
\end{equation}
Since $\{\ket{\mathbf{k}}| \mathbf{k}\in \mathbb{Z}_{+,0}^{\infty}\}$ forms a complete orthonormal basis for $\Gamma_s(\mathcal{H})$, $\{pe^{-\sum\limits_{j=1}^{N} s_jk_j}\big| k_j \in \mathbb{Z}_{+}, N\in \mathbb{N}\}$ is the complete set of eigen values for $\rho_s$.
If $\{s_j\}$ is linearly independent over $\mathbb{Q}$, then we see that the eigenvalues corresponding to $\ket{\mathbf{k}_1} \neq \ket{\mathbf{k}_2}$ are not same. Thus the multiplicity of each of these eigenvalues is one. 
\end{proof}

\begin{lem}\label{lem:symm-group-1}
Let $\rho_s$ be as in Lemma \ref{lem:symm-group} where  $\{s_j\}_{j\in \mathbb{N}}$ is a  sequence of (distinct) irrational numbers which are linearly independent over the field $\mathbb{Q}$. Then a unitary operator $U$ in $\GH$ is such that $U\rho_sU^*$ is a Gaussian state if and only if for some $\alpha \in \CH$, $L\in \SH \CH$ and a complex valued function $\beta$ of modulus one on $\mathbb{Z}_{+,0}^{\infty}$
\begin{equation*}
  U = W(\alpha)\Gamma_s(L)\beta(a_1^{\dagger}a_1,a_2^{\dagger}a_2,\dots),
\end{equation*}
 where $\beta(a_1^{\dagger}a_1,a_2^{\dagger}a_2,\dots)$ is the unique unitary which satisfies \[ \beta(a_1^{\dagger}a_1,a_2^{\dagger}a_2,\dots)\ket{\mathbf{k}} = \beta(\mathbf{k})\ket{\mathbf{k}}, \forall \mathbf{k}\in \mathbb{Z}_{+,0}^{\infty}.\]
\end{lem}
\begin{proof}
  Since $\beta(a_1^{\dagger}a_1,a_2^{\dagger}a_2,\dots)$ commutes with $\rho_s$ the sufficiency is immediate.
  To prove the necessity, suppose $U\rho_sU^*$ is Gaussian.
 Then by Theorem \ref{thm:struct-quant-gauss} there exists $z\in \mathcal{H}$,  $M\in \mathbf{S}(\mathcal{H})$ and $\rho_t:= \Pi_{j=1}^{\infty}(1-e^{-t_j})\otimes_{j=1}^{\infty}e^{-t_ja_{j}^{\dagger}a_j} \in \B{\Gamma_s(\mathcal{H})}$ such that
\begin{equation}
\label{eq:88}
U\rho_sU^*= W(z)^{*}\Gamma_s(M)^{*}\rho_t\Gamma_s(M)W(z).
\end{equation}
 By Lemma \ref{lem:symm-group},  $\rho_s$ has a complete orthonormal eigenbasis with distinct eigenvalues.  By (\ref{eq:88}) $\rho_s$ and $\rho_t$ are unitarily equivalent and thus their eigenvalues and multiplicities are same. Therefore by applying Lemma \ref{lem:symm-group} to $\rho_t$ , $\Pi_{j=1}^{\infty}(1-e^{-t_j})=p$ (since $p$ is the maximum eigenvalue of $\rho_s$) and $\rho_t$ has a set of distinct eigenvalues $pe^{-\sum\limits_{j=1}^{N} t_jk_j}$ corresponding to the eigenvectors $\ket{\mathbf{k}}, \mathbf{k} = (k_1 ,k_2, \dots, k_N, 0,0,\dots)^{\tau}\in  \mathbb{Z}_{+,0}^{\infty}$, $N\in \mathbb{N}$. 
\begin{claim}
 The sequence $\{t_j\}_{j\in \mathbb{N}}$ consists of (distinct) numbers which are linearly independent over the field $\mathbb{Q}$.
\end{claim}
\begin{pf}[of Claim]
If $t_i = t_k$ for some $i\neq k$ then it is possible to choose distinct $\mathbf{k}, \mathbf{k'} \in \mathbb{Z}^{\infty}_{+,0}$  such that the eigenvalues of $\rho_t$ corresponding to $\ket{\mathbf{k}}$ and $\ket{\mathbf{k'}}$ are same. This will imply that the corresponding eigenspace is at least two dimensional which is not possible. 
To see the rational independence note that for any two finite subsets $I, J \subset \mathbb{N}$, $\sum\limits_{j\in I} t_jk_j \neq \sum\limits_{j\in J} t_jk_j^{'}$ where $k_j, k_{j}^{'} \in \mathbb{Z_+}, \forall j$. Now if 
\begin{equation}
\label{eq:91}
\sum\limits_{j=1}^Nt_jq_j = 0
\end{equation} for a finite collection of rational numbers $q_j$'s, since $t_j>0,\forall j$ then there must be negative rational numbers in the set $\{q_1,q_2, \dots,q_N\}$ (unless  $q_j=0, \forall j$). Then \ref{eq:91} can be written in the form $\sum\limits_{j\in I} t_jk_j = \sum\limits_{j\in J} t_jk_j^{'}$ for two finite sets $I,J$, which is not possible. Thus the claim is proved.
\end{pf}
We have $\{pe^{-\sum\limits_{j=1}^N s_jk_j}\big| k_j \in \mathbb{Z}_{+}, N \in \mathbb{N}\} = \{pe^{-\sum\limits_{j=1}^N t_jk_j}\big| k_j \in \mathbb{Z}_{+}, N \in \mathbb{N}\}$. Therefore  
   $\{ \sum\limits_{j=1}^ns_jk_j \vert k_j\in \mathbb{Z}_+ \forall j, n\in \mathbb{N}\} = \{ \sum\limits_{j=1}^nt_jk_j \vert k_j\in \mathbb{Z}_+\forall j, n\in \mathbb{N}\}$. Now by  the proof of Lemma \ref{lem:basic-technical},  there is a bijection
   $\sigma : \mathbb{N}\to \mathbb{N}$ such that $s_j=t_{\sigma (j)}$ for all $j\in \mathbb{N}.$ 

By (\ref{eq:88}) there exists a unitary $V$ such that 
\begin{equation}
\label{eq:89}
V\rho_sV^{*} = \rho_t.
\end{equation}
where $V = \Gamma_s(M)W(z)U$.
Let $\mathbf{k} = (k_1 ,k_2, \dots, k_N, 0,0,\dots)^{\tau}\in  \mathbb{Z}_{+,0}^{\infty}$ be arbitrary, by \ref{eq:89} if $\ket{\mathbf{k}}$ is an eigenvector for $\rho_s$ then $V\ket{\mathbf{k}}$ is an eigenvector for $\rho_t$ with the same eigenvalue. Therefore, $V\ket{\mathbf{k}}$ is an eigenvector for $\rho_t$ with eigenvalue $pe^{-\sum\limits_{j=1}^N s_jk_j}=
 pe^{-\sum _{j\in \mathbb{N}}t_{\sigma {j}}k_j}.$ 
But defining the unitary operator  $A$ on $\mathcal{H}$,  $A_\sigma (e_j)= e_{\sigma (j)}$, its second quantization, $\Gamma (A) \mathbf{k} $ is an eigenvector for $\rho _t$ with same eigenvalue. Since the multiplicity for each eigenvalue is one,  there exists a complex number $\beta(\mathbf{k})$ of unit modulus such that,  
\begin{align*}
V\ket{\mathbf{k}}&= \beta(\mathbf{k})\ket{A\mathbf{k}}\\
                 &= \Gamma_s(A)\beta(\mathbf{k})\ket{\mathbf{k}}\\
                 &= \Gamma_s(A)\beta(a_1^{\dagger}a_1,a_2^{\dagger}a_2,\dots)\ket{\mathbf{k}}
\end{align*}
Then by (\ref{eq:89}) $U = W(z)^{*}\Gamma_s(M)^{*}\Gamma_s(A)\beta(a_1^{\dagger}a_1,a_2^{\dagger}a_2,\dots)$. Now the proof is complete due to Theorem \ref{thm:generalized-shale-unitary}. It should be noted that we may need to redefine $\beta$ if the multiplier $\sigma(M^{-1}, A)\neq 1$  (refer Theorem \ref{thm:generalized-shale-unitary}).  
\end{proof} 

The following theorem characterizes the Gaussian symmetries. 
Recall the definition of Shale operators from \ref{sec:sympl-autom}. 
\begin{thm}\label{thm:Gaussian-symmetry}
  A unitary operator $U \in \B{\GH}$ is a Gaussian symmetry if and only if
  \begin{equation*}
    U = \lambda W(\alpha)\Gamma_s(L)
  \end{equation*}
for some $\lambda\in \mathbb{C}$ with $\abs{\lambda} = 1$, $\alpha\in \CH$, and $L$ is a Shale operator ($L \in \mathbf{S}(\mathcal{H})$).
\end{thm}
\begin{proof}
  The sufficiency is immediate. 
  To prove the necessity,  let us consider $\CH = \oplus_j\mathbb{C}e_j$ with respect to some orthonormal basis $\{e_{j}\}$. By considering a Gaussian state  $\rho_s$ as in Lemma \ref{lem:symm-group-1},  we can assume without loss of generality that $U =\beta(a_1^{\dagger}a_1,a_2^{\dagger}a_2,\dots)$ for some $\beta .$ We will show that $U = \Gamma_s(D)$ for some unitary operator $D$ and this will prove the theorem because of (\ref{item:22}) of Theorem \ref{thm:generalized-shale-unitary}.

Let $\psi \in \GH$ be such that $\ketbra{\psi}{\psi}$ is a pure Gaussian state. Then by assumption $\ketbra{U\psi}{U\psi}$ is also a Gaussian state and it is clearly a pure state as it is obtained from the wave function $\ket{U\psi}$. We choose  the coherent state (Example \ref{eg:coherent-state})
 \[\psi = e^{-\frac{1}{2}\|u\|^2}\ket{e(u)}= W(u)\ket{e(0)}, \] 
where $u = \sum _{j}u_je_j$.Now 
\begin{equation}
\label{eq:93}
\ket{U\psi}= e^{-\frac{1}{2}\|u\|^2}\beta(a_1^{\dagger}a_1,a_2^{\dagger}a_2,\dots)\ket{e(u)}
\end{equation}
Recall that 
$$e(u)= \sum\limits_{\mathbf{k}\in \mathbb{Z}_{+,0}^{\infty}}^{}\frac{u^{\mathbf{k}}}{\sqrt{\mathbf{k}!}}\ket{\mathbf{k}}$$
Consider any `finite' vector $z= \sum _{j=1}^Nz_je_j$, for some $z_1, z_2, \ldots , z_N$ in $\mathbb{C}, N\in \mathbb{N}.$ Then \begin{equation}
\label{eq:96}
e(z)=\sum\limits_{\mathbf{m}\in \mathbb{Z}_+^N}\frac{z^{\mathbf{m}}}{\sqrt{\mathbf{m}!}}\ket{\mathbf{m}},
\end{equation}
where $\mathbf{m}\in \mathbb{Z}_+^N$ is considered as the vector $(m_1,m_2,\dots,m_N,0,0,\dots)^t\in\mathbb{Z}_{+,0}^{\infty} $ and $\ket{\mathbf{m}} = \ket{m_1}\ket{m_2}\cdots\ket{m_n} \in \GH$ as in the notation of (\ref{eq:85}).
We will evaluate the function $f(z)= \left\langle U\psi, e(z) \right\rangle$ using (\ref{eq:93}) and (\ref{eq:94}). From  (\ref{eq:101}), (\ref{eq:93}),  (\ref{eq:96}) and continuity of $\beta(a_1^{\dagger}a_1,a_2^{\dagger}a_2,\dots)$ we have

\begin{align*}
f(z)&=e^{-\frac{1}{2}\|u\|^2}\left\langle Ue(u), e(z) \right\rangle\\
&=e^{-\frac{1}{2}\|u\|^2} \left\langle  \sum\limits_{\mathbf{k}\in \mathbb{Z}_{+,0}^{\infty}}^{}\frac{u^{\mathbf{k}}}{\sqrt{\mathbf{k}!}}U\ket{\mathbf{k}}, \sum\limits_{\mathbf{m}\in \mathbb{Z}_+^N}\frac{z^{\mathbf{m}}}{\sqrt{\mathbf{m}!}}\ket{\mathbf{m}}  \right\rangle\\
&= e^{-\frac{1}{2}\|u\|^2}\left\langle  \sum\limits_{\mathbf{k}\in \mathbb{Z}_{+,0}^{\infty}}^{}\frac{u^{\mathbf{k}}}{\sqrt{\mathbf{k}!}}\beta(\mathbf{k})\ket{\mathbf{k}},  \sum\limits_{\mathbf{m}\in \mathbb{Z}_+^N}\frac{z^{\mathbf{k}}}{\sqrt{\mathbf{k}!}}\ket{\mathbf{m}} \right\rangle \numberthis \label{eq:95.1}
\end{align*}
Thus 
\begin{equation}\label{eq:99}
f(z) =e^{-\frac{1}{2}\|u\|^2}\sum\limits_{\mathbf{k}\in \mathbb{Z}_+^{N}}\frac{(\bar{u}_1z_1)^{k_1}(\bar{u}_2z_2)^{k_2}\cdots(\bar{u}_Nz_N)^{k_N}}{k_{1}!k_2!\cdots k_N!}\overline{\beta(\mathbf{k})}.
\end{equation}

Since $\abs{\beta(\mathbf{k})} = 1$, from (\ref{eq:99}) we see that 
\begin{equation}
\label{eq:100}
\abs{f(z)}\leq \exp{-\frac{1}{2}\|u\|^2+\sum\limits_{j=1}^N\abs{u_j}\abs{z_j}}.
\end{equation}

Now we compute $f(z)$ in an alternative way. As $|U\psi \rangle \langle U\psi |$ is a pure Gaussian state, 
by Corollary \ref{cor:struct-pure-quant-gauss} there exists a unitary $A$ and an $\alpha\in \mathcal{H}$ such that 
\begin{equation}
\label{eq:94}
\ket{U\psi}= W(\alpha) \Gamma_s(A)\otimes_j\ket{e_{\lambda_j}}.
\end{equation}

From the definition of $e(w)$ and $e_{\lambda}$ in $\Gamma(\mathbb{C})=L^2(\mathbb{R})$, we have 
$$e_{\lambda}(x) = (2\pi)^{-1/4}\lambda^{-1/2}\exp{-4^{-1}\lambda^{-2}x^2}, \phantom{...} x \in \mathbb{R}, \lambda >0$$ on $L^2(\mathbb{R})$. Further, 

\begin{equation}
\label{eq:98}
\left\langle e_{\lambda},e(w)  \right\rangle = \sqrt{\frac{2\lambda}{1+\lambda^2}}\exp \frac{1}{2}\left( \frac{\lambda^2-1}{\lambda^2+1} \right)w^2, \lambda >0, w \in \mathbb{C}.
\end{equation}
Using (\ref{eq:94}),
\begin{align*}
f(z)&= \left\langle  W(\alpha) \Gamma_s(A)\otimes_je_{\lambda_j}, e(z)\right\rangle\\
&= \left\langle \otimes_je_{\lambda_j}, \Gamma_s(A^*)W(-\alpha)e(z) \right\rangle\\
&= e^{\langle \alpha,z \rangle -\frac{1}{2}\|\alpha\|^{2}}\left\langle \otimes_je_{\lambda_j}, e\big(A^{*}(z-\alpha)\big)  \right\rangle \numberthis \label{eq:100.1}
\end{align*}
Since $z$ is a finite vector and $\alpha$ is fixed, each coordinate of $A^{*}(z-\alpha)$ is a first degree polynomial in $z_j$'s.  Therefore $ e\big(A^{*}(z-\alpha)\big) = \otimes_j e(w_j)$ where each $w_j$ is a first degree polynomial in $z_j$'s. Therefore from (\ref{eq:98}) and \todo[noline]{details needed? as both entries in the i.p are infinite tensor products} property of infinite tensor products 
\[ f(z) = e^{\langle \alpha,z \rangle -\frac{1}{2}\|\alpha\|^{2}} \lim_{n\rightarrow\infty}\Pi_{j=1}^{n}\sqrt{\frac{2\lambda_{j}}{1+\lambda_j^2}}\exp \sum\limits_{j=1}^n\frac{1}{2}\left( \frac{\lambda_j^2-1}{\lambda_j^2+1} \right)w_j^2 \]

Since each $w_j^2$ is a second degree polynomial in $z_1,z_2,\dots,z_N$. This contradicts (\ref{eq:100}) unless $\lambda_j = 1$ for all $j$. Now (\ref{eq:94}) implies 
\begin{align*}
\ket{U\psi}&= W(\alpha) \Gamma_s(A)\ket{e(0)}\\
 &=  e^{-\frac{1}{2}\|\alpha\|^2}\ket{e(\alpha)}
\end{align*}
Now from (\ref{eq:93}) we get 
\begin{equation}
\label{eq:103}
e^{-\frac{1}{2}\|u\|^2}\beta(a_1^{\dagger}a_1,a_2^{\dagger}a_2,\dots)\ket{e(u)}= e^{-\frac{1}{2}\|\alpha\|^2}\ket{e(\alpha)}
\end{equation}
Thus $\beta(a_1^{\dagger}a_1,a_2^{\dagger}a_2,\dots)$ is a unitary with the following properties: 
\begin{enumerate}
\item\label{item:23} $\beta(a_1^{\dagger}a_1,a_2^{\dagger}a_2,\dots) \ket{\mathbf{k}}= \beta(\mathbf{k})\ket{\mathbf{k}}$ for every $\mathbf{k}\in \mathbb{Z}_{+,0}^{\infty}$.
\item \label{item:24} It maps coherent vectors to coherent vectors.
\end{enumerate}
We will prove that  $\beta(a_1^{\dagger}a_1,a_2^{\dagger}a_2,\dots)= \Gamma_s(D)$ for a diagonal unitary $D$. To this end we fix a  $u = \sum _ju_je_j$ in $\mathcal{H}$   with $u_j \neq 0, \forall j$. We have $ \beta(a_1^{\dagger}a_1,a_2^{\dagger}a_2,\dots)\ket{e(u)} = e^{\frac{1}{2}(\|u\|^2-\|\alpha\|^2)}\ket{e(\alpha)}$. Therefore if $\alpha = \sum _j\alpha _je_j$ from (\ref{eq:103}) and (\ref{eq:101}) we get, 
\[\sum\limits_{\mathbf{k}\in \mathbb{Z}_{+,0}^{\infty}}^{}\frac{u^{\mathbf{k}}}{\sqrt{\mathbf{k}!}}\beta(\mathbf{k})\ket{\mathbf{k}} =  e^{\frac{1}{2}(\|u\|^2-\|\alpha\|^2)} \sum\limits_{\mathbf{k}\in \mathbb{Z}_{+,0}^{\infty}}^{}\frac{\alpha^{\mathbf{k}}}{\sqrt{\mathbf{k}!}}\ket{\mathbf{k}}\]
Therefore,
 \[u^{\mathbf{k}}\beta(\mathbf{k}) = e^{\frac{1}{2}(\|u\|^2-\|\alpha\|^2)} \alpha^{\mathbf{k}},  \forall \mathbf{k}\in \mathbb{Z}_{+,0}^{\infty}\]
Since $u_{j}\neq 0$ for all $j$, we see that if $\mathbf{k}= (k_1,k_2,\dots,k_m,0,0,\dots)\in \mathbb{Z}_{+,0}^{\infty}$, 
\begin{equation*}
\beta(\mathbf{k}) =  e^{\frac{1}{2}(\|u\|^2-\|\alpha\|^2)} \left(\frac{\alpha_1}{u_1} \right)^{k_1} \left(\frac{\alpha_2}{u_2} \right)^{k_2}\cdots \left(\frac{\alpha_m}{u_m} \right)^{k_m}, \forall \mathbf{k}\in \mathbb{Z}_{+,0}^{\infty}
\end{equation*}
Since $\abs{\beta(\mathbf{k})}= 1$, we get $\abs{\frac{\alpha_j}{u_j}}= 1$ for all $j$. If we write $\frac{\alpha_j}{u_j} = e^{i\theta_j}$, then from (\ref{eq:103}) we get 
\begin{equation}
\label{eq:105}
\beta(a_1^{\dagger}a_1,a_2^{\dagger}a_2,\dots)\ket{e(u)}= \ket{e(Du)},
\end{equation} where $D$ is the unitary $\diag(e^{i\theta_1}, e^{i\theta_2}, \dots)$ , for every $u= \sum _ju_je_j \in \mathcal{H}$   with $u_j\neq 0, \forall j$. Now it is easy to see that (\ref{eq:105}) holds for all $u\in \mathcal{H}$. We conclude  that $\beta(a_1^{\dagger}a_1,a_2^{\dagger}a_2,\dots) = \Gamma_s(D)$.
\end{proof}
\begin{rmk}
Observe that the set of all Gaussian symmetries $\{\lambda W(\alpha)\Gamma_s(L): \lambda\in \mathbb{C}, \abs{\lambda}=1, \alpha \in \CH, L \in \SH \CH  \}$ form a group under multiplication. A typical computation in the proof of the previous statement is given below. 
\begin{align*}
\lambda_1W(\alpha_1)\Gamma_s(L_1)\lambda_2W(\alpha_2)\Gamma_s(L_2) 
&= \lambda_1\lambda_2W(\alpha_1)\Gamma_s(L_1)W(\alpha_2)\Gamma_s(L_1)^{-1}\Gamma_s(L_1)\Gamma_s(L_2)\\
&= \lambda_1\lambda_2W(\alpha_1)W(L_1\alpha_2)\Gamma_s(L_1L_2)\\
&= \lambda_1\lambda_2e^{-i \im \langle\alpha_1, L_1\alpha_2\rangle}W(\alpha_1+L_1\alpha_2)\Gamma_s(L_1L_2).
\end{align*}
The quantity in the last line is again in the form of a prototypical element in the set of Gaussian symmetries.
\end{rmk}
 
 Combining Theorem \ref{thm:Gaussian-symmetry} with Theorem \ref{thm:struct-quant-gauss}, we have the following important observation.
\begin{rmk}
The symplectic spectrum (Remark \ref{rmk:symplectic-spectrum}) of the covariance operator is a complete invariant for  Gaussian states and two Gaussian states with same symplectic spectrum are conjugate to each other through a Gaussian symmetry.
\end{rmk}

\textbf{Acknowledgements:} We wish to thank Prof. K. R. Parthasarathy for introducing the subject to us during his 2015 lecture series at the Indian Statistical Institute, Bangalore centre. John thanks Prof. Martin Lindsay and Prof. K. B. Sinha for discussions on quasifree states and metaplectic representations respectively.  Bhat thanks J C Bose Fellowship for financial support. We sincerely thank the anonymous referee for several constructive suggestions which helped us to improve the paper immensely. 

\bibliographystyle{amsalpha}
\bibliography{bibliography-2}
\end{document}